\documentclass{article}
\usepackage{graphicx}
\usepackage{latexsym}
\usepackage{amsmath,amssymb,amstext}
\usepackage{comment}
\usepackage{pifont}
\usepackage{paralist}

\usepackage{amsthm}
\usepackage{enumerate}
\usepackage{color}
\usepackage{wrapfig}
\usepackage{mathtools}

\usepackage{indent}
\usepackage{atbeginend}

\usepackage{comment}

\usepackage{tikz}
\usetikzlibrary{arrows,automata,backgrounds,shapes.callouts,calc,chains,fadings,folding,decorations.fractals,decorations.pathreplacing,fit,patterns,positioning,mindmap,shadows,shapes.geometric,shapes.symbols,through,trees,plotmarks}

\usepackage{xcolor}
\usepackage{framed}
\definecolor{shadecolor}{named}{lightgray}

\renewcommand{\mit}{\mathit}

\newcommand{\seqn}{{=_{\mit{nf}}}}
\newcommand{\seqh}{{=_{\mit{hnf}}}}
\newcommand{\seqw}{{=_{\mit{whnf}}}}
\newcommand{\eqn}{\mathrel{\seqn}}
\newcommand{\eqh}{\mathrel{\seqh}}
\newcommand{\eqw}{\mathrel{\seqw}}

\newcommand{\extended}[1]{#1}
\renewcommand{\pagebreak}{}
\newcommand{\short}[1]{}

\theoremstyle{plain}
\newtheorem{definition}{Definition}

\newtheorem{theorem}{Theorem}
\newtheorem{lemma}{Lemma}

\newtheorem{remark}{Remark}
\newtheorem{notation}{Notation}

\makeatletter
\def\moverlay{\mathpalette\mov@rlay}
\def\mov@rlay#1#2{\leavevmode\vtop{%
   \baselineskip\z@skip \lineskiplimit-\maxdimen
   \ialign{\hfil$#1##$\hfil\cr#2\crcr}}}
\makeatother

\newcommand{\picitem}[1]{
  \scalebox{0.8}{
  \vspace{.5ex}
  \begin{tikzpicture}[baseline=-1ex]
    \node [circle,drop shadow,fill=blue!10,draw,outer sep=.7mm,inner sep=.5mm] {#1};
  \end{tikzpicture}\,}
  }

\newcommand{\appsection}[1]{
  \let\oldthesection\thesection
  \renewcommand{\thesection}{Appendix \oldthesection}
  \section{#1}
  \let\thesection\oldthesection
}

\renewcommand{\theenumi}{(\roman{enumi})}
\renewcommand{\labelenumi}{\theenumi}

\newcommand{\llbracket}{{[\![}}
\newcommand{\rrbracket}{{]\!]}}

\newcommand{\ignore}[1]{}

\newcommand{\mcl}{\mathcal}
\newcommand{\msf}{\mathsf}

\newcommand{\pairlft}{{\langle}}                
\newcommand{\pairrgt}{{\rangle}}                
\newcommand{\pairsep}{{,\,}}                    
\newcommand{\pairstr}[1]{\pairlft#1\pairrgt}    
\newcommand{\pair}[2]{\pairstr{#1\pairsep#2}}   
\newcommand{\triple}[2]{\pair{#1\pairsep#2}}    %
\newcommand{\quadruple}[2]{\triple{#1\pairsep#2}} 

\newcommand{\tuple}[1]{\pairstr{#1}}   
\newcommand{\sfunin}{{:}}
\newcommand{\funin}{\mathrel{\sfunin}}
\newcommand{\powerset}[1]{\wp(#1)}

\newcommand{\setemp}{{\varnothing}}

\newcommand{\setsize}{\lstlength}

\newcommand{\mybind}[3]{#1#2.\:#3}
\newcommand{\myex}{\mybind{\exists}}
\newcommand{\myall}{\mybind{\forall}}

\newcommand{\mylam}[2]{\lambda#1.#2}

\newcommand{\nat}{\mathbb N}

\catcode`\@=11

\newcommand{\mywash}[2]{\setbox0=\hbox{$\m@th#1{#2}$}\wd0=0pt\box0}
\catcode`\@=12

\newcommand{\asig}{\Sigma}

\newcommand{\avars}{\mathcal{X}}
\newcommand{\ster}{\mit{Ter}}
\newcommand{\ter}{\bfunap{\ster}}
\newcommand{\TT}{\ter{\Sigma}{\avars}}
\newcommand{\TZ}{\ter{\Sigma}{\setemp}}

\newcommand{\atrs}{R}

\newcommand{\binap}[3]{#2\mathbin{#1}#3}
\newcommand{\funap}[2]{#1(#2)}

\newcommand{\bfunap}[3]{\funap{#1}{#2,#3}}

\newcommand{\qfunap}[5]{\funap{#1}{#2,#3,#4,#5}}
\newcommand{\pfunap}[6]{\funap{#1}{#2,#3,#4,#5,#6}}

\newcommand{\where}{\mathrel{|}}

\newcommand{\sdefdby}{{:}{=}}
\newcommand{\defdby}{\mathrel{\sdefdby}}
\newcommand{\defd}{\defdby}

\renewcommand{\implies}{\Rightarrow}

\newcommand{\aalg}{\mathcal{A}}

\newcommand{\xlb}{\scalebox{1}[.9]{$\llbracket$}}
\newcommand{\xrb}{\scalebox{1}[.9]{$\rrbracket$}}
\newcommand{\sint}{\xlb{\cdot}\xrb}
\renewcommand{\int}[1]{\xlb #1 \xrb}

\newcommand{\sorts}{\mcl{S}}

\newcommand{\autstates}{Q}

\newcommand{\astate}{q}

\newcommand{\lang}{\funap{\mcl{L}^\omega}}

\newcommand{\strff}{\msf} 

\newcommand{\lstemp}{\varepsilon} 

\newcommand{\lstlength}[1]{|#1|}

\newcommand{\astr}{\sigma}
\newcommand{\bstr}{\tau}

\newcommand{\sstrcns}{\strff{:}}

\newcommand{\strcnsd}[1]{\binap{\sstrcns}}

\newcommand{\hole}{\Box}
\newcommand{\cxthole}{\hole}

\newcommand{\cpi}[2]{\mathrm{\Pi}^{#1}_{#2}}
\newcommand{\csig}[2]{\mathrm{\Sigma}^{#1}_{#2}}
\newcommand{\cdel}[2]{\mathrm{\Delta}^{#1}_{#2}}

\newcommand{\tm}{\msf{M}}
\newcommand{\tmstates}{Q}
\newcommand{\tmsig}{\Gamma}
\newcommand{\stmtrans}{\delta}
\newcommand{\tmtrans}{\bfunap{\stmtrans}}

\newcommand{\stmblank}{\Box}

\newcommand{\tmL}{L}
\newcommand{\tmR}{R}

\newcommand{\tmstart}{\astate_0}

\newcommand{\tmes}[1]{\aes_{#1}}
\newcommand{\tmesn}[1]{\aes^{\mit{n}}_{#1}}
\newcommand{\tmtrs}[1]{\atrs_{#1}}
\newcommand{\tmtrsn}[1]{\atrs^{\mit{n}}_{#1}}

\newcommand{\pto}{\rightharpoonup}

\newcommand{\pp}[2]{#1|_{#2}}
\newcommand{\reducible}{\le}
\newcommand{\domain}[1]{\textit{domain}(#1)}
\newcommand{\lencode}[1]{\langle\!\langle #1 \rangle\!\rangle}

\newcommand{\sbitstream}{\sstream}
\newcommand{\sstream}{\mit{S}}
\newcommand{\sbit}{\mit{B}}
\newcommand{\snat}{\mit{N}}
\newcommand{\ones}{\msf{ones}}
\newcommand{\zeros}{\msf{zeros}}
\newcommand{\alt}{\msf{blink}}
\newcommand{\oftype}{\mathrel{{:}{:}}}
\newcommand{\ar}[1]{\mit{ar}(#1)}

\newcommand{\scons}{:}
\newcommand{\cons}{\mathrel{\scons}}

\newcommand{\sunary}{\msf{unary}}
\newcommand{\unary}{\funap{\sunary}}
\newcommand{\siszeros}{\msf{is}_{\msf{zeros}}}
\newcommand{\iszeros}{\funap{\siszeros}}

\newcommand{\suhd}{\msf{uhd}}
\newcommand{\uhd}{\funap{\suhd}}
\newcommand{\sutl}{\msf{utl}}
\newcommand{\utl}{\funap{\sutl}}

\newcommand{\suc}{\funap{\msf{s}}}

\newcommand{\shd}{\msf{head}}
\newcommand{\hd}{\funap{\shd}}
\newcommand{\stl}{\msf{tail}}
\newcommand{\tl}{\funap{\stl}}
\newcommand{\szip}{\msf{zip}}
\newcommand{\zip}{\bfunap{\szip}}
\newcommand{\szipn}[1]{\szip_{#1}}
\newcommand{\zipn}[1]{\funap{\szip_{#1}}}
\newcommand{\sdup}{\msf{dup}}
\newcommand{\dup}{\funap{\sdup}}
\newcommand{\sinv}{\msf{inv}}
\newcommand{\inv}{\funap{\sinv}}
\newcommand{\seven}{\msf{even}}
\newcommand{\even}{\funap{\seven}}

\newcommand{\snatstr}{\msf{natstr}}
\newcommand{\natstr}{\funap{\snatstr}}

\newcommand{\xbis}{\equiv}

\newcommand{\oracle}{\xi}

\newcommand{\aes}{E}

\newcommand{\sstart}{\msf{S}}

\newcommand{\srun}{\msf{run}}

\newcommand{\nb}{\nobreakdash}
\newcommand{\sdots}{\scalebox{.9}{\ldots}}

\newcommand{\sol}[2]{\int{#1}_{#2}}
\newcommand{\solful}[2]{\int{#1}_{#2\!,\,\text{full}}}

\newcommand{\num}{\underline}
\newcommand{\satisfies}{\mid\!\hspace{-0.25pt}\equiv}

\newcommand{\problem}[2]{{\setlength{\leftmargini}{47pt}
  \vspace{-1ex}
  \begin{itemize}
    \item [{\sc Input:}] #1
    \item [{\sc Question:}] #2
  \end{itemize}\vspace{-1ex}}}

\newcommand{\tminit}[2]{\tm(#1 ; #2)}
\newcommand{\inttminit}[2]{\int{\tm}(#1 ; #2)}

\AfterBegin{enumerate}{}
\AfterBegin{definition}{\normalfont}
\AfterBegin{remark}{\normalfont}

\newcommand{\levi}{L\'{e}vy--Longo}
\newcommand{\levy}{\levi}
\newcommand{\ber}{Berarducci}
\newcommand{\bohm}{B\"{o}hm}

\newcommand{\slesseq}{\msf{leq}}
\newcommand{\lesseq}{\bfunap{\slesseq}}

\newcommand{\srtSetA}{C}
\newcommand{\srtSetB}{D}

\newcommand{\tiff}{\;\;\text{iff}\;\;}

\begin{document}
\newcommand{\empytset}{\setemp}

\title{%
  On the Complexity of Equivalence 
  of Specifications of Infinite Objects
}

\author{%
  J\"{o}rg Endrullis \and Dimitri Hendriks \and Rena~Bakhshi \\
  VU University Amsterdam \\
  Department of Computer Science \\
  De Boelelaan 1081a \\
  1081 HV Amsterdam \\
  The Netherlands \\
  \texttt{\{j.endrullis,\,r.d.a.hendriks,\,r.r.bakhshi\}@vu.nl}
}

\maketitle

\begin{abstract}
    We study the complexity of deciding the equality of infinite objects
  specified by systems of equations, and of infinite objects specified by $\lambda$-terms.
  For equational specifications there are several natural notions of equality:
  equality in all models,
  equality of the sets of solutions, and
  equality of normal forms for productive specifications.
  For $\lambda$-terms we investigate B\"ohm-tree equality 
  and various notions of observational equality.
  We pinpoint the complexity of each of these notions
  in the arithmetical or analytical hierarchy.
  
  We show that the complexity of deciding equality in all models subsumes the entire analytical hierarchy.
  This holds already for the most simple infinite objects, 
  viz.\ streams over 
  $\{0,1\}$, 
  and stands in sharp contrast to the low arithmetical $\cpi{0}{2}$-completeness 
  of equality of equationally specified streams derived in~\cite{rosu:2006}
  employing a different notion of equality.

\end{abstract}

%

\section{Introduction}\label{sec:intro}
In the last two decades interest has grown towards infinite data,
as witnessed by the application of type theory to infinite objects~\cite{coqu:1993}, 
as well as the emergence of coalgebraic techniques for infinite~data types
like streams~\cite{rutt:2003}, 
infinitary term rewriting and infinitary lambda calculus~\cite{tere:2003}.
In functional programming, the use of 
infinite data structures
dates back to 1976, see~\cite{hend:morr::1976,frie:wise:1976}.

We are concerned with the complexity of deciding the equality of 
infinite objects specified by systems of equations,
and infinite objects specified by $\lambda$-terms.
The equational specification of coinductive objects is common practice in
coalgebra, term rewriting and functional programming. 
Consider the following example from~\cite{rosu:2006}:
\begin{gather}
  \left.
  \begin{aligned}
    \zeros &= 0:\zeros &&&
    \ones &= 1:\ones\\
    \alt &= 0:1:\alt &&&
    \zip{x:\sigma}{\tau} &= x:\zip{\tau}{\sigma}
  \end{aligned}
  \hspace{0.3cm}\right\}\hspace{-0.3cm}
  \label{spec:zip:alt}
\end{gather}
This is an equational specification of three infinite lists of bits, 
and a binary function over infinite lists.%
  \footnote{%
    In Haskell there is $\szip \oftype {[a] \to [b] \to [(a,b)]}$,
    but we prefer to use `zip' 
    for the \emph{interleaving} of lists,
    as defined by the equation in \eqref{spec:zip:alt},
    since that is what a zipper does: 
    it interleaves rows of teeth.%
  }
Then, a typical question is whether the following equality holds: 
\begin{align}
  \zip{\zeros}{\ones} = \alt \label{eq:zip:alt}
\end{align}
%
The answer depends on the semantics we choose to interpret the equality;
for example \eqref{eq:zip:alt} is not valid in the hidden models considered in~\cite{rosu:2006};
for more details we refer to Section~\ref{sec:related}.
In order to answer such a question, 
we first need to settle on the precise semantics of equality
for equational specifications;
the candidates we consider in this paper are
{%
\renewcommand{\theenumi}{\Roman{enumi}}
\renewcommand{\labelenumi}{\theenumi.}
\begin{enumerate}
  \item\label{eq:model} Equality in all models.
  \item\label{eq:solution}\label{eq:uniquex} Equality of the set of solutions.
\end{enumerate}%
}
For $\lambda$-terms we are not concerned with equality in the sense of convertibility
(which is known to be $\cpi{0}{2}$-complete, see~\cite{bare:1984}).
Instead, we are interested in \emph{behavioral} equivalence of $\lambda$-terms in all contexts,
because this corresponds to the interchangeability of expressions in purely functional languages. 
It is also closely related to referential transparency, 
and the notion of B\"ohm trees as values of expressions including those without normal form.
Thus we consider the following equivalences for $\lambda$-terms:
{%
\renewcommand{\theenumi}{\Roman{enumi}}
\renewcommand{\labelenumi}{\theenumi.}
\begin{enumerate}
  \setcounter{enumi}{2}
  \item\label{eq:obs} Observational equivalences. 
  \item\label{eq:bohm} B\"ohm-tree equality.
\end{enumerate}%
}
The `right' choice of equivalence depends on the intended application.
The classic semantics mentioned in items~\ref{eq:model} and \ref{eq:uniquex} above, 
are defined by model-theoretic means.
From a algebraic perspective \ref{eq:model} and \ref{eq:uniquex} 
are the most natural semantics to consider for equational reasoning.
On the other hand, \ref{eq:obs} and \ref{eq:bohm}, 
are defined by means of evaluation, i.e., rewriting.
In functional programming the latter are of foremost importance,
because these take (lazy) evaluation strategies into account.
From an evaluation perspective, two terms are equal 
if they have the same observable behavior,
independent of the context they are in.
In contrast to the model-theoretic notions, this equality is invariant under the exchange of 
\emph{meaningless subterms}, that is, subterms which cannot be evaluated to a (weak) head normal form.

Another candidate for the semantics of equality is
{%
\renewcommand{\theenumi}{\Roman{enumi}}
\renewcommand{\labelenumi}{\theenumi.}
\begin{enumerate}
  \setcounter{enumi}{4}
  \item\label{eq:prod} Equality of normal forms for productive specifications.
\end{enumerate}%
}
\noindent
A rewrite specification is \emph{productive}~\cite{sijt:1989,endr:grab:hend:2008} 
if the terms under consideration can be fully evaluated, that is, (outermost-fair) rewriting 
yields a (possibly infinite) constructor normal form in the limit.
In such a setting, equality of the normal forms 
is a suitable semantics for the equivalence of terms. 
Deciding the equality of productive specifications 
has been shown to be a $\cpi{0}{1}$\nb-complete problem 
in~\cite{grab:endr:hend:klop:moss:2012}; 
this semantics is not considered here.

We now briefly describe the concepts \ref{eq:model}--\ref{eq:bohm}.

\paragraph{Equality in models (\ref{eq:model} and \ref{eq:solution}).}
The semantics~\ref{eq:model} (equality in all models) is useful when the objects
under consideration are specified in the same specification.
This semantics interprets the objects simultaneously in each model satisfying the specification.
This allows us to compare objects that depend on a common unknown, an underspecified object;
see~\eqref{spec:X} below for an illustrating example.
If the objects under consideration are fully specified, that is, have unique solutions,
then semantics~\ref{eq:model} coincides with semantics~\ref{eq:solution}.

In contrast to~\ref{eq:model}, semantics~\ref{eq:solution} is more suitable 
for comparing objects specified by different specifications, as we explain below.
The objects are compared via the set of their solutions (in their respective specifications).
This semantics is well-known from equations over real (or complex) numbers,
where two equations, like
\begin{align*}
  (x-1)^2 - 1 &= 0 &&\text{ and }& x^2-2x&= 0\,,
\end{align*}
are equivalent if they have the same solutions for $x$, here $\{0,2\}$.


A $\asig$-algebra $\aalg$ consists of a carrier set $A$ (the domain of $\aalg$) 
and an interpretation $\sint$ of the symbols $\asig$ occurring in the equational specification 
as functions over $A$.
Then $\aalg$ is called a \emph{model}~of~an equational specification $E$, 
which we denote by $\aalg \models E$,
if all equations of $E$ respect the interpretation; 
that is, for every equation of $E$ both sides have the same interpretation 
for every assignment of the variables.
As the domain we will typically choose (a subset of) the \emph{final coalgebra}~\cite{sang:rutt:2012}
describing the class of objects we are specifying.
The final coalgebra ensures that the model is \emph{continuous},
that is, if we have a converging sequence of terms $t_1,t_2,\ldots$ with limit~$t_\omega$,
then the sequence of interpretations $\int{t_1},\int{t_2},\ldots$ converges towards $\int{t_\omega}$.
For example, in a specification like
\begin{align}
  \ones &= 1:\ones &
  \ones' = 1:\ones'
\end{align}
the symbols $\ones$ and $\ones'$ are guaranteed to have the same interpretation.
Continuity is crucial to conclude the validity of equations such as $\ones = \ones'$
which are not satisfied in non-continuous models like the \emph{initial} algebra of the specification.
%

Let $E$ be a specification of $M$ and $N$.
Then $M$ is considered equal to $N$ with respect to semantics \ref{eq:model} 
if every model of $E$ is also a model of $M = N$: \hspace{2mm}
  $\myall{\aalg}{ \;\;\aalg \models E \;\;\implies\;\; \aalg \models M = N }\;$. 
%
This notion is especially of interest if $M$ and $N$ depend on a common unknown 
and consequently have to be interpreted simultaneously in the same model.
For example in
\begin{gather}
  \left.
  \begin{aligned}
  M &= \zip{X}{X} & \zip{x:\sigma}{\tau} &= x:\zip{\tau}{\sigma} \\
  N &= \dup{X} & \dup{x:\sigma} &= x:x:\dup{\sigma} \,
  \end{aligned}
  \hspace{0.2cm}\right\}\hspace{-.3cm}\label{spec:X}
\end{gather}
the streams $M$ and $N$ are both specified in terms of an unspecified stream~$X$.
Whatever interpretation $X$ has, $M$ and $N$ are equal,
and so they are equal in the sense of semantics~\ref{eq:model}.

On the other hand, semantics~\ref{eq:model} has the effect that an underspecified constant 
is not equivalent to its renamed copy. 
This is illustrated by the following specification:
\begin{align}
  M &= 0:\tl{M} & 
  N &= 0:\tl{N} \label{spec:copy}
\end{align}
Here $M$ and $N$ are not equal in every model;
for example, let $\int{M} = 0:0:\ldots$ and $\int{N} = 0:1:1:\ldots$.
Nevertheless, $M$ and $N$ are equal in the sense that 
they exhibit the same behaviors. 
That is, they have the same set of solutions:
every stream starting with a zero is a solution for $M$ as well as for $N$.
Thus, $M$ and $N$ are equal with respect to the semantics~\ref{eq:solution}. 
This paves the way for comparing objects $M$ and $N$ 
that are given by separate specifications $E_M$ and $E_N$, respectively.
Note that it is not always suitable to apply semantics~\ref{eq:model}
to the union $E_M \cup E_N$ even if the specifications have disjoint signatures (using renaming),
see further Remark~\ref{rem:union}.

%
Two objects $M$ and $N$ are equal with respect to semantics~\ref{eq:solution} if
the set of solutions of $M$ in $E_M$ coincides with the set of solutions of $N$ in $E_N$:
  \;\;$\{\, \int{M}^{\aalg} \mid \aalg \models E_M \,\}
  \;=\; \{\, \int{N}^{\aalg} \mid \aalg \models E_N \,\}$\;.
Here the set of solutions of a constant $X$ in a specification $E_X$ 
is the set of interpretations of $X$ in all models of $E_X$.

\paragraph{Observational equivalence (\ref{eq:obs} and \ref{eq:bohm}).}
In purely functional languages based on the $\lambda$-calculus~\cite{bare:1984}, 
the evaluation of expressions is \emph{free of side effects}.
As a consequence, an expression (or subexpression) can always be replaced by its normal form, 
the so-called \emph{value} of the expression. 
This principle is known as \emph{referential transparency}.
This also implies that expressions can be substituted for each other 
if they have the same normal form.

For specifications of coinductive objects, such as infinite lists (called \emph{streams}) or infinite trees,
the value typically is an infinite term.
For example in $\ones = 1:\ones$,
the term $\ones$ has as~value (or infinite normal form) the infinite term $1:1:1:\ldots$.
However, it is not always guaranteed that a term can be fully evaluated.
During the evaluation to the (possibly infinite) normal form, we may encounter subterms that cannot be evaluated 
because these subterms do not have a head normal~form.
In $\lambda$\nb-calculus, such terms are known as \emph{meaningless terms}.
For example, consider:
\begin{align*}
  \msf{natsx}(n) &= n : \msf{g}(0) : \msf{natsx}(n+1) & \msf{g}(n) = \msf{g}(n)\\
  \msf{natsx'}(n) &= n : \msf{g}(n) : \msf{natsx'}(n+1) 
\end{align*}
Here $\msf{g}(n)$ is meaningless for every $n$. 
Consequently, $\msf{natsx}(0)$ evaluates to a stream in which every second element is meaningless, and therefore, 
undefined. An infinite value containing undefined parts can be represented by means of B\"ohm trees~\cite{bare:1984} 
introduced in 1975 by Corrado B\"ohm.
In particular, the B\"ohm tree of $\msf{natsx}(0)$~is:
$0:\bot:1:\bot:2:\bot:3:\bot:4:\bot:\ldots$,
where $\bot$ is a special symbol representing an undefined element.

In $\lambda$-calculus (or orthogonal higher-order rewriting), 
terms with equal B\"ohm trees can be exchanged (for each other) 
without changing the meaning of the whole expression.
In the specification above, 
$\msf{natsx}(0)$ and $\msf{natsx'}(0)$ have the same B\"ohm tree,
and hence are interchangeable. 
In contrast, from the model-theoretic perspective $\msf{natsx}(0)$ and $\msf{natsx'}(0)$ are different.
In every model of $\msf{natsx}(0)$ all elements at odd indexes coincide, 
whereas $\msf{natsx'}(0)$ admits models that assign different interpretations to these elements.
From a rewriting as well as functional programming perspective, these differences are irrelevant 
as they concern undefined subterms.

There are several notions of infinite values,
depending on what terms are considered meaningless,
including B\"ohm trees, L\'evy-Longo trees, Berarducci trees, 
$\eta$-B\"ohm trees, $\eta^\infty\!\!$-B\"ohm trees; see further~\cite{deza:giov:2001}.
The terms $\mylam{x}{xx}$ and $\mylam{x}{x(\mylam{z}{xz})}$, for instance, 
have distinct B\"ohm trees, but we may want to consider the terms
\emph{behaviorally}, or \emph{observationally equivalent} as they are $\eta$-convertible.
There are several natural concepts of \emph{observational equivalence} for $\lambda$-calculus, 
where terms are considered \emph{equivalent} if they yield the same observations in every context. 
To that end, we consider three forms of \emph{observations}: 
normal forms (nf), head normal forms (hnf), and weak head normal forms (whnf).
A \emph{head normal form} is a $\lambda$\nb-term
of the form $\mylam{x_1}{\ldots\mylam{x_n}{y N_1 \ldots N_m}}$ with $n, m \ge 0$.
A \emph{weak head normal form} is an hnf or an abstraction,
i.e., a whnf is a term of the form $x M_1 \ldots M_m$ or $\mylam{x}{M}$.
Each of the observations gives rise to an equivalence $\seqn$, $\seqh$ or $\seqw$, defined by
{\allowdisplaybreaks
\begin{align*}
  M \eqn N &\tiff (\,\myall{C}{C[M] \text{ has a nf} \tiff C[N] \text{ has a nf}}\,)\\
  M \eqh N &\tiff (\,\myall{C}{C[M] \text{ has a hnf} \tiff C[N] \text{ has a hnf}}\,)\\
\pagebreak
  M \eqw N &\tiff (\,\myall{C}{C[M] \text{ has a whnf} \tiff C[N] \text{ has a whnf}}\,)
\end{align*}}%
In fact, the equivalence $\seqn$ corresponds to $\eta$-B\"ohm trees,
and $\seqh$ to $\eta^\infty\!\!$\nb-B\"ohm trees.
For more details we refer to~\cite{deza:giov:2001}, where it is argued 
that $\seqw$ corresponds to evaluation strategies used by lazy functional languages.
If two expressions behave the same in every context, then no functional program can distinguish them. 



\paragraph{Contribution.}
We characterize for each of the semantics \ref{eq:model}--\ref{eq:bohm} 
the complexity of deciding the equality of terms.
For \ref{eq:model} and \ref{eq:solution} we will focus on equational specifications of bitstreams,
and for \ref{eq:obs} and \ref{eq:bohm} on behavioral equivalences of $\lambda$-terms and B\"ohm tree equality.

Each of these equivalences is undecidable, therefore we characterize their complexity
by means of the arithmetical and analytical hierarchies,
see Figure~\ref{fig:hierarchy}.
The arithmetical hierarchy classifies the complexity of a problem $P$ by the minimum 
number of quantifier alternations in first-order formulas that characterize $P$.
The analytical hierarchy extends this classification
to second-order arithmetic, then counting the alternations of set quantifiers.

{
\newcommand{\mdown}{-1mm}
\renewcommand{\small}{}
\begin{figure}[!pb]
  \vspace{-2ex}
  \begin{center}
  \scalebox{0.9}{
  \begin{tikzpicture}[node distance=12mm,inner sep=.5mm,thick,>=stealth]
    \node (d10) {$ \cpi{1}{0} = \cdel{1}{0} = \csig{1}{0}$};
    \node (p11) [above left of=d10,yshift=\mdown] {$\cpi{1}{1}$}; \draw [->] (d10) -- (p11);
    \node (s11) [above right of=d10,yshift=\mdown] {$\csig{1}{1}$}; \draw [->] (d10) -- (s11);
    \node (d12) [above right of=p11,yshift=\mdown] {$\cdel{1}{2}$}; \draw [->] (p11) -- (d12); \draw [->] (s11) -- (d12);
    \node (p12) [above left of=d12,yshift=\mdown] {${\cpi{1}{2}}$}; \draw [->] (d12) -- (p12);
    \node (s12) [above right of=d12,yshift=\mdown] {$\csig{1}{2}$}; \draw [->] (d12) -- (s12);
    \node (d13) [above right of=p12,yshift=\mdown] {$\vdots$}; \draw [->] (p12) -- (d13); \draw [->] (s12) -- (d13);
  
    \begin{scope}[every node/.style={rectangle,outer sep=2mm,inner sep=.5mm}]
    \node [anchor=west] at (d10.east) {{\small \emph{arithmetic predicates $\Phi$}}};
    \node [anchor=west] at (s11.east) {{\small $\myex{X}{\Phi(X)}$}};
    \node [anchor=east] (fp11) at (p11.west) {{\small $\myall{X}{\Phi(X)}$}};
    \node [anchor=west] at (s12.east) {{\small $\myex{X}{\myall{Y}{\Phi(X,Y)}}$}};
    \node [anchor=east] at (p12.west) {{\small $\myall{X}{\myex{Y}{\Phi(X,Y)}}$}};
    \end{scope}
    
    \node (d00) [below of=d10,node distance=35mm] {$ \cpi{0}{0} = \cdel{0}{0} = \csig{0}{0}$};
    \node (p01) [above left of=d00,yshift=\mdown] {$\cpi{0}{1}$}; \draw [->] (d00) -- (p01);
    \node (s01) [above right of=d00,yshift=\mdown] {$\csig{0}{1}$}; \draw [->] (d00) -- (s01);
    \node (d02) [above right of=p01,yshift=\mdown] {$\cdel{0}{2}$}; \draw [->] (p01) -- (d02); \draw [->] (s01) -- (d02);
    \node (p02) [above left of=d02,yshift=\mdown] {${\cpi{0}{2}}$}; \draw [->] (d02) -- (p02);
    \node (s02) [above right of=d02,yshift=\mdown] {$\csig{0}{2}$}; \draw [->] (d02) -- (s02);
    \node (d03) [above right of=p02,yshift=\mdown] {$\vdots$}; \draw [->] (p02) -- (d03); \draw [->] (s02) -- (d03);
  
    \begin{scope}[every node/.style={rectangle,outer sep=2mm,inner sep=.5mm}]
    \node [anchor=west] at (d00.east) {{\small \emph{decidable predicates $D$}}};
    \node [anchor=west] (fs01) at (s01.east) {{\small $\myex{x}{D(x)}$}};
    \node [anchor=east] at (p01.west) {{\small $\myall{x}{D(x)}$}};
    \node [anchor=west] at (s02.east) {{\small $\myex{x}{\myall{y}{D(x,y)}}$}};
    \node [anchor=east] (fp02) at (p02.west) {{\small $\myall{x}{\myex{y}{D(x,y)}}$}};
    \end{scope}

    \begin{scope}[every node/.style={rectangle,outer sep=2mm,inner sep=.5mm}]
      \node [anchor=south east,yshift=-4mm] at (fp11.north east) {{\small \emph{well-foundedness}}};
      \node [anchor=south east,yshift=-4mm] at (fp02.north east) {{\small \emph{totality}}};
      \node [anchor=south west,yshift=-4mm] at (fs01.north west) {{\small \emph{recursively enumerable\hspace*{.5cm}}}};
    \end{scope}

    \begin{scope}[every node/.style={circle,drop shadow,fill=blue!10,draw,outer sep=1mm,inner sep=.5mm},node distance=5mm]
      \node [above of=d13] {A};
      \node [anchor=south,xshift=-.5mm,yshift=.0mm] at (p12.north) {C};
      \node [anchor=south,xshift=-.5mm,yshift=.0mm] at (p11.north) {B};
      \node [anchor=north,xshift=0mm,yshift=.5mm] at (d12.south) {E};
      \node [anchor=south,xshift=-.5mm,yshift=.0mm] at (p02.north) {D};
    \end{scope}
  \end{tikzpicture}}
  \end{center}
  \caption{Arithmetical (bottom) and analytical hierarchy (top).}
  \label{fig:hierarchy}
\end{figure}
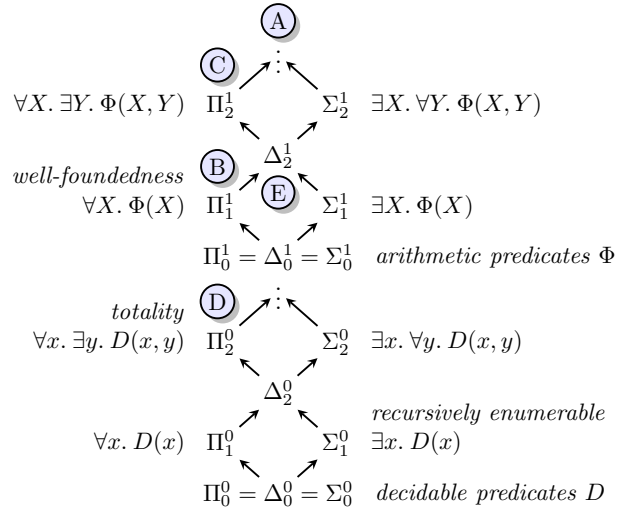
}
%
\picitem{A} It turns out that the complexities of deciding the equality in all models as well as the  
equality of the set of solutions  
subsume the entire arithmetical and analytical hierarchy
when the domain of the models is the set of all streams,
so-called \emph{full} models, see Theorems~\ref{thm:models:full} and~\ref{thm:solutions:full}.
The idea of the proof is as follows.
We translate formulas of the analytical hierarchy into stream specifications
by representing $\forall$ set quantifiers by equations with variables.
This simulates a quantification over all streams as the models are full,
and the equations have to hold for all assignments of the variables.
The $\exists$ set quantifiers are eliminated in favor of Skolem functions (here stream functions).
The interpretation of the functions is determined by the model,
and the question whether there exists a model corresponds 
to an existential quantification over all Skolem functions.

\picitem{B}\!{\&}\!\!\picitem{C}
If we admit models whose domain does not contain all streams,
then the complexity of deciding equality drops to the level $\cpi{1}{1}$ of the analytical hierarchy
for semantics~\ref{eq:model}, and to $\cpi{1}{2}$ for~\ref{eq:solution},
see Theorems~\ref{thm:models}, and~\ref{thm:solutions}.
The reason is that equations with variables no longer have to hold for all streams,
but only for the streams that exist in the model. 
By the L\"owenheim-Skolem theorem we obtain that if
there exists a model, then there exists a countable model:
from an uncountable model we construct
a countable one, by taking the finitely many streams ``of~interest'' 
and closing them under all functions in the model. 
\pagebreak
Thus, it suffices to quantify over countable models
for which one single set quantifier is enough.

The aforementioned results already hold for bitstreams, 
one of the simplest coinductive objects,
and thereby can serve as a lower bound on the hardness of the equality problem for other coinductive objects.
We also study the behavioral semantics from~\cite{rosu:2006}.
We find that if behavioral equivalence $\equiv$ is required to be a congruence,
like for example in~\cite{bido:henn:kurz:2003},
then the complexity of deciding behavioral equivalence is
catapulted out of the arithmetical hierarchy,
to the level $\cpi{1}{1}$ of the analytical hierarchy, see Theorem~\ref{thm:behavioral:models}.
Likewise so for the
behavioural equivalence for specifications of 
streams of natural numbers, relaxing the congruence requirement,
see Theorem~\ref{thm:confusion}.


\picitem{D} For the equivalences on $\lambda$\nb-terms, we show that deciding the 
\bohm{} tree and \levi{} tree equality, as well as
the observational equivalences $\seqn$, $\seqh$ and $\seqw$ 
are $\cpi{0}{2}$-complete problems, see Theorem~\ref{thm:lambda}.
(It is clear that when an object is given by a rewrite system, 
like the $\lambda$\nb-calculus, then the complexity resides in the arithmetical hierarchy,
since it suffices to quantify over a number steps to evaluate parts of the object.)

\picitem{E} Finally, we consider the complexity of unique solutions.
A term $s$ has a unique solution within a specification $E$ if 
there exists models of $E$, and in all models of $E$, 
$s$ has the same interpretation.
The problem of deciding unique solvability in all full models 
again subsumes the analytical hierarchy, see \short{Theorem~\ref{thm:all:full}}\extended{Theorems~\ref{thm:atleast:full}, \ref{thm:unique:full} and~\ref{thm:atmost:full}}.
When also considering the non\nb-full models, we find that the problem
is $\cpi{1}{1}$- and $\csig{1}{1}$-hard, but is strictly contained in $\cdel{1}{2}$,
see \short{Theorem~\ref{thm:unique}}\extended{Theorems~\ref{thm:atmost}, \ref{thm:atleast} and~\ref{thm:unique}}.
  

\paragraph*{Outline.} We first discuss related work.
We formally introduce bitstream specifications and stream models  
in Section~\ref{sec:specs}, and Turing machines with oracles 
in Section~\ref{sec:turing}. We recall the basic complexity-related notions in Section~\ref{sec:levels}.
We use these concepts in Section~\ref{sec:models:theoretic}~to derive the complexity results for the
model-theoretic notions.  
In Section \ref{sec:rosu} we consider a different notion of models, 
namely the behavioural semantics as in~\cite{rosu:2006}. 
Finally, we investigate the observational equivalences of $\lambda$-terms in Section~\ref{sec:lambda}.

\section{Related Work}\label{sec:related}

The complexity of the equality of streams specified by systems of equations
has been investigated in the ICFP paper~\cite[Corollary 1]{rosu:2006}; we cite:
\emph{Proving equality on streams defined equationally is a $\cpi{0}{2}$-complete problem}.
This result is based on a behavioral notion of stream models~\cite{rosu:2000}.
We briefly summarize the main characteristics of these models:
\begin{enumerate}
  \item Every stream $\sigma \in \{0,1\}^\omega$ 
    can have multiple representatives in the model (known as \emph{confusion}).
  \item For every equation $\ell = r$ it is required that
    the interpretations $\int{\ell}$ and $\int{r}$ are \emph{behaviorally equivalent},
    denoted by $\equiv$, that is, equality under all $\int{\shd}(\int{\stl}^n(\hole))$ experiments.
    In particular, it is not required that $\int{\ell} = \int{r}$.
  \item Behavioral equivalence $\equiv$ is not required to be a congruence.
\end{enumerate}
Behavioral models have a wide range of applications,
for example for modeling computations with hidden states,
or capturing certain forms of nondeterminism.
For these applications it is often intended that the semantics is not preserved under equational reasoning.
For example, consider the following specification from~\cite{rosu:2006}
\newcommand{\spush}{\msf{push}}
\newcommand{\push}{\funap{\spush}}
\begin{align*}
  \tl{\push{\sigma}} = \sigma\;,
\end{align*}
specifying a function $\spush$ that prefixes an element to the argument stream,  
while leaving unspecified which element.
In the behavioral models we obtain a restricted form of nondeterminism~\cite{wali:meld:2001},
for example, the following is not behaviorally satisfied:
\begin{align}
  \push{\tl{\push{\sigma}}} = \push{\sigma}\;,
  \label{eq:push}
\end{align}
although derivable by pure equational reasoning.
For a nondeterministic operation, 
it is of course desirable that \eqref{eq:push} does not hold.

However, for function definitions employing pattern matching,
behavioral models sometimes yield unexpected results; consider:
\begin{align}
  \ones &= 1:\ones &
  \msf{f}(x:\sigma) &= \sigma \label{spec:ones:f}
\end{align}
Now, there are models that satisfy the specification~\eqref{spec:ones:f},
but not~\eqref{eq:f:ones}:
\begin{align}
  \msf{f}(\ones) = \ones\,
  \label{eq:f:ones}
\end{align}
In these models we have that $\int{\ones} \ne \int{1\,{:}\,\ones}$ 
and, at the same time, that $\int{\ones}$
cannot be constructed by the stream constructor~$\int{{:}}$,
that is, $\int{\ones} \ne \int{{:}}(x,s)$ for all $x \in \{0,1\}$ and~$s \in A_{\sbitstream}$.
Consequently, the interpretation $\int{\msf{f}}(\int{\ones})$ can be arbitrary.

\extended{%
\begin{remark}
  We construct a behavioral model for specification~\ref{spec:ones:f} 
  in which $\msf{f}(\ones) = \ones\,$ is not satisfied:
  {\allowdisplaybreaks
  \renewcommand{\u}{}
  \begin{gather*}
    A_{\sbitstream} = \{o\} \cup \{\u{w} \where w \in \{0,1\}^\omega\} \quad \text{(the domain of the model)}\\[0ex]
    \begin{aligned}
    \int{\shd}(o) &= 1 &
    \int{\shd}(\u{0w}) &= 0 &
    \int{\shd}(\u{1w}) &= 1 \\
    \int{\stl}(o) &= o &
    \int{\stl}(\u{0w}) &= \u{w} &
    \int{\stl}(\u{1w}) &= \u{w} \\[0ex]
    \int{{:}}(0,o) &= \u{0\,(1^\omega)} &
    \int{{:}}(0,w) &= \u{0w} &
    \int{{:}}(1,w) &= \u{1w} \\
    \int{{:}}(1,o) &= \u{1^\omega} & \int{\ones} &= o\\[0ex]
    \int{\msf{f}}(o) &= \u{0^\omega} &
    \int{\msf{f}}(\u{0w}) &= \u{w} &
    \int{\msf{f}}(\u{1w}) &= \u{w}
    \end{aligned}
  \end{gather*}
  }%
  This model illustrates that the requirements of~\cite{rosu:2006} 
  do not ensure that every element of the stream domain $A_{\sbitstream}$ can be constructed by the stream constructor $\int{{:}}$.
  For example, the element $o$ represents the stream of ones, 
  but $o \ne \int{{:}}(a,b)$ for all $a \in \{0,1\}$ and~$b \in A_{\sbitstream}$.
  In general, $\int{M}$ and $\int{\hd{M} : \tl{M}}$
  need not be the same element of the domain, although they are behaviorally equivalent.
\end{remark}}

Thus, behavioral reasoning is typically not sound for behavioral models,
and therefore the corresponding specifications are usually referred to as \emph{behavioral specifications}.
In this paper we are interested in specifications where equational reasoning is sound.

\begin{remark}\label{counterexample}

We construct a behavioral model $\pair{A}{\sint{}}$ in the sense of~\cite{rosu:2006} 
where specification~\eqref{spec:zip:alt}
is behaviorally satisfied but the goal equation $\zip{\zeros}{\ones} = \alt$ is not.
The model thereby forms a counterexample to \cite[Example 2]{rosu:2006}.

We define the domain by $A_{\sbit} = \{0,1\}$ and
\begin{gather*}
  A_{\sbitstream} = \{z_w \mid w \in \{0,1\}^*\} \cup \{o_w \mid w \in \{0,1\}^*\} \cup \{0,1\}^\omega
\end{gather*}
Here $z_\lstemp$ and $o_\lstemp$ are alternative representations of $0^\omega$ and $1^\omega$,
respectively, and $z_w$ and $o_w$ have an additional finite prefix $w \in \{0,1\}^*$.
We define the interpretations $\sint{}$
for every $a \in \{0,1\}$, $\sigma \in \{0,1\}^\omega$, $w,v \in \{0,1\}^*$ and $x,y \in A_{\sbitstream}$.
For $\int{\shd}$ and $\int{\stl}$ we define:
\begin{align*}
  \int{\shd}(z_{\varepsilon}) &= 0 &
  \int{\shd}(o_{\varepsilon}) &= 1 &
  \int{\shd}(a\sigma) &= a \\
  \int{\stl}(z_\varepsilon) &= z_\varepsilon &
  \int{\stl}(o_\varepsilon) &= o_\varepsilon &
  \int{\stl}(a\sigma) &= \sigma \\  
  \int{\shd}(z_{aw}) &= a &
  \int{\shd}(o_{aw}) &= a \\
  \int{\stl}(z_{aw}) &= z_w &
  \int{\stl}(o_{aw}) &= o_w 
\end{align*}
We define the interpretation $\int{{:}}$ of the stream constructor by:
\begin{align*}
  \int{{:}}(a,z_w) &= z_{aw} &
  \int{{:}}(a,o_w) &= o_{aw} &
  \int{{:}}(a,\sigma) &= a\sigma
\end{align*}
Note that the elements $z_\lstemp$ and $o_\lstemp$ cannot be constructed by $\int{{:}}$.

We interpret $\int{\zeros}$, $\int{\ones}$ and $\int{\alt}$ as follows: 
\begin{align*}
  \int{\zeros} &= z_\varepsilon &
  \int{\ones} &= o_\varepsilon &
  \int{\alt} &= (01)^\omega &
\end{align*}
We define an auxiliary function $\Join$ that (similar to zip) 
interleaves the elements of finite or infinite words; 
for $u_1,u_2 \in \{0,1\}^{\le \omega} = \{0,1\}^* \cup \{0,1\}^\omega$, let
  $au_1 \Join u_2 = a(u_2 \Join u_1)$ and
  $\lstemp \Join u_2 = u_2$.
%
We now define the interpretation $\int{\szip}$ of the symbol $\szip$ as follows:
\begin{align*}
  \int{\szip}(z_w,o_v) &= 
  \begin{cases}
    (w \Join v) 0^\omega &\text{for $|w| = |v|$}\\
    (w\;0^\omega) \Join (v\;1^\omega) &\text{otherwise}
  \end{cases}\\
  \int{\szip}(o_w,z_v) &= 
  \begin{cases}
    (w \Join v) 0^\omega &\text{for $|w| = |v|+1$}\\
    (w\;1^\omega) \Join (v\;0^\omega) &\text{otherwise}
  \end{cases}
\end{align*}
\newcommand{\emb}{\funap{\mit{emb}}}%
and in all other cases, we define
$\int{\szip}(x,y) = \emb{x} \Join \emb{y}$
where 
$\emb{z_w} = w 0^\omega$,\,
$\emb{o_w} = w 1^\omega$, and
$\emb{\sigma} = \sigma$.

\short{%
  It is straightforward to check that specification~\ref{spec:zip:alt} is behaviorally satisfied,
  whereas
  $\int{\zip{\zeros}{\ones}} = 0^\omega \not\equiv (01)^\omega = \int{\alt}$.
}

\extended{
We check that the specification~\ref{spec:zip:alt} is behaviorally satisfied:
\begin{align*}
  \int{\zeros} = z_\varepsilon &\equiv z_0 = \int{0:\zeros}\\
  \int{\ones} = o_\varepsilon &\equiv o_1 = \int{1:\ones}\\
  \int{\alt} = (01)^\omega &= 01 (01)^\omega = \int{0:1:\alt}
\end{align*}
Observe that $z_\lstemp$ and $z_0$ (and likewise $o_\lstemp$ and $o_1$) are behaviorally equivalent. 
For the $\szip$ equation we distinguish the following cases:
\begin{enumerate}
  \item\label{pr:i} If $\int{\sigma} = z_w$, $\int{\tau} = o_v$, we have:
    \begin{align*}
      \int{\zip{x:\sigma}{\tau}} &= 
      \begin{cases}
        (xw \Join v) 0^\omega &\text{if $|xw| = |v|$}\\
        (xw\;0^\omega) \Join (v\;1^\omega) &\text{otherwise}
      \end{cases}\\
      \int{x:\zip{\tau}{\sigma}} &=
      \begin{cases}
        x (v \Join w) 0^\omega &\text{if $|v| = |w|+1$}\\
        x ((v\;1^\omega) \Join (w\;0^\omega)) &\text{otherwise}
      \end{cases}
    \end{align*}
    The equality $\int{\zip{x:\sigma}{\tau}} = \int{x:\zip{\tau}{\sigma}}$
    follows by the definition of $\Join$
    together with $|xw| = |w|+1$.

  \item The case $\int{\sigma} = o_v$, $\int{\tau} = z_w$ is analogous to \ref{pr:i}.
    
  \item If $\int{\sigma},\int{\tau} \in \{0,1\}^\omega$, we have:
      \begin{align*}
        \int{\zip{x:\sigma}{\tau}} &= (x\int{\sigma}) \Join \int{\tau} 
        \\ &= x(\int{\tau} \Join \int{\sigma}) 
        = \int{x:\zip{\tau}{\sigma}}
      \end{align*}

  \item\label{pr:iv} If $\int{\sigma} = z_w$, $\int{\tau} \in \{0,1\}^\omega$, then:
      \begin{align*}
        \int{\zip{x:\sigma}{\tau}} &= (xw\;0^\omega) \Join \int{\tau} \\
        &= x (\int{\tau} \Join (w\;0^\omega)) = \int{x:\zip{\tau}{\sigma}} 
      \end{align*}
  \item The case $\int{\sigma} = o_w$, $\int{\tau} \in \{0,1\}^\omega$ is analogous to~\ref{pr:iv}.
  \item The case $\int{\sigma} \in \{0,1\}^\omega$, $\int{\tau} = z_w$ is analogous to~\ref{pr:iv}.
  \item The case $\int{\sigma} \in \{0,1\}^\omega$, $\int{\tau} = o_w$ is analogous to~\ref{pr:iv}.
  \item\label{pr:viii} If $\int{\sigma} = z_w$, $\int{\tau} = z_v$, then:
      \begin{align*}
      \int{\zip{x:\sigma}{\tau}} &= (xw\;0^\omega) \Join (v\;0^\omega) \\
                   &= x((v\;0^\omega) \Join (w\;0^\omega)) = \int{x:\zip{\tau}{\sigma}} 
      \end{align*}
  \item The case $\int{\sigma} = o_w$, $\int{\tau} = o_v$ is analogous to~\ref{pr:viii}.

\end{enumerate}
Hence $\pair{A}{\sint}$ behaviorally satisfies Specification~\ref{spec:zip:alt}. 
However:
\begin{align*}
  \int{\zip{\zeros}{\ones}} &= \int{\szip}(\int{\zeros},\int{\ones}) \\
  &= \int{\szip}(z_\varepsilon,o_\varepsilon) = 0^\omega
\end{align*}
whereas
\begin{align*}
  \int{\alt} &= (01)^\omega
\end{align*}
Consequently, the equation $\zip{\zeros}{\ones} = \alt$ is not behaviorally satisfied in this model. 
}

%

  %
\end{remark}

%


The counterexample in Remark~\ref{counterexample} employs the fact that
the behavioral models of~\cite{rosu:2006} do not require
that every stream can be constructed by the (interpretation of the) stream constructor $\int{{:}}$.
As a consequence, the equation $\zip{x:\sigma}{\tau} = x:\zip{\tau}{\sigma}$
does not fully define $\int{\szip}$;
it defines $\int{\szip}(\sigma,\tau)$
only for those arguments $\sigma$ that can be constructed by $\int{{:}}$.

The example illustrates
that the behavioral models of~\cite{rosu:2006} do not go along with 
function definitions using pattern matching.
To fully define $\int{\szip}$, we can specify it using the stream destructors:
$\hd{\zip{\sigma}{\tau}} = \hd{\sigma}$, and
$\tl{\zip{\sigma}{\tau}} = \zip{\tau}{\tl{\sigma}}$.
This  change of the specification format resolves the problem.

Alternatively, keeping the specification format,
we can adapt the notion of models.
To reestablish soundness of equational reasoning one can
(i)~exclude confusion
or~(ii) require that $\equiv$ is a congruence.
Note that the common models of streams are free of confusion:
final coalgebras~\cite{rutt:2005b}, 
one-sided infinite words $A^\omega$, 
and the function space $\nat \to A$.
%
%
%
%
In hidden algebras~\cite{malc:1997}, confusion is often allowed
but its negative effects are prevented by restricting to
behavioral models~\cite{bido:henn:kurz:2003}, in which behavioral equivalence
is a congruence: 
$s \xbis t \implies f(\ldots,s,\ldots) \xbis f(\ldots,t,\ldots)$.
\extended{
Then equational reasoning is sound with respect to behavioral equality,
and for a specification like
$\ones = 1:\ones$,
$\ones' = 1:\ones'$,
the equality $\msf{g}(\ones) = \msf{g}(\ones')$ holds behaviorally.
}%
%
%


Our results show that when $\xbis$ is required to be a congruence (or confusion is eliminated),
then the complexity of the equality of bitstreams that are specified equationally
jumps from the low level $\cpi{0}{2}$ of the arithmetical hierarchy to
the level $\cpi{1}{1}$ of the analytical hierarchy, thereby exceeding
the arithmetical hierarchy. 
Moreover, we show that even for behavioral specifications with confusion
(as in~\cite{rosu:2006}), equality of \emph{streams of natural numbers} is $\cpi{1}{1}$-complete.
Consequently, the results of~\cite{rosu:2006} are valid only for bitstreams 
in combination with the behavioral equality discussed above.
For general behavioral specifications (not the special case of stream specifications), the $\cpi{1}{1}$-completeness 
has been shown in~\cite{buss:rosu:2000}.
%

Term rewriting systems are closely related to equational specifications.
The complexity of deciding various standard properties of term rewriting systems, such as
productivity, termination and confluence (Church--Rosser),
has been investigated in~\cite{endr:geuv:simo:zant:2011,endr:grab:hend:2009}. 


\section{Bitstream Specifications}\label{sec:specs}

We will focus mainly on \emph{streams}, one-sided infinite sequences of symbols,
the prime example of coinductive structures.
%
%
There are various ways of introducing streams:
as functions $\nat \to A$ mapping an index $n$ to the $n$-th element of the stream,
as final coalgebras  over the functor $X \mapsto A\times X$,
using coinductive types~\cite{geuv:1992},
or observational models~\cite{bido:henn:kurz:2003}.
All these definitions are equivalent in the sense
that the resulting coalgebras are isomorphic. 

For the model-theoretic semantics of equality,
we will focus on specifications of bitstreams, 
streams 
over the alphabet $\{0,1\}$.
Due to their simplicity, bitstreams can be embedded 
in almost every non-trivial coinductive structure. 
Specifications of bitstreams are inherently sorted, 
with a sort $\sbit$ for bits, and a sort $\sbitstream$ for bitstreams.
To this end, we introduce sorted terms.
Let $\sorts$ be a set of sorts; 
an \emph{$\sorts$-sorted set} $\srtSetA$ is a family of sets $\{{\srtSetA}_s\}_{s\in \sorts}$.
Let $\srtSetA$ and $\srtSetB$ be $\sorts$-sorted sets.
Then an $\sorts$-sorted \emph{function} (or \emph{map}) from $\srtSetA$ to $\srtSetB$
is a function $f \funin \srtSetA \to \srtSetB$ such that
$\funap{f}{{\srtSetA}_s} \subseteq {\srtSetB}_s$ for all $s \in \sorts$, 
that is, a function that respects the sorts.

An $\sorts$-sorted signature $\asig$ is a set of symbols $f \in \asig$,
each having a type $(s_1,\ldots,s_n,s) \in \sorts^{n+1}$,
denoted by $f \oftype s_1 \times \ldots \times s_n \to s$, where $n$ is the arity of $f$.
Let~$\avars$ be an $\sorts$-sorted set of \emph{variables}.
The $\sorts$-sorted set of \emph{terms} $\TT$ is inductively defined by:
\begin{itemize}
  \item $\avars_s \subseteq \TT_s$ for every $s \in \sorts$, and 
  \item $f(t_1,\ldots,t_n) \in \TT_s$ if $f \oftype s_1 \times \ldots \times s_n \to s$,
    $f \in \asig$, and $t_1 \in \TT_{s_1},\ldots,t_n \in \TT_{s_n}$.
\end{itemize}
An $\sorts$-sorted \emph{equation $\ell = r$} consists of terms $\ell,r \in \TT_s \times \TT_s$ for some $s\in\sorts$.

\begin{definition}\normalfont
  A \emph{bitstream signature $\asig$} is an $\sorts$-sorted signature with $\sorts = \{\sbit,\sbitstream\}$
  such that $0,1,\,{:} \in \asig$ where
  $0,1 \oftype \sbit$ are called \emph{bits}, and
  the infix symbol `${:}$' of type $\sbit \times \sbitstream \to \sbitstream$ is the \emph{stream constructor}.
  An \emph{equational bitstream specification} over $\asig$ is a \emph{finite} set $\aes$ 
  of equations over $\asig$.
\end{definition}

From now on we let $\sorts = \{\sbit,\sbitstream\}$.

\begin{definition}\normalfont
  A \emph{stream algebra} $\aalg = \pair{A}{\sint}$ consists of:
  \begin{enumerate}
    \item an $\sorts$-sorted domain $A$; $A_{\sbit} = \{0,1\}$ and
          $\setemp \ne A_{\sbitstream} \subseteq \{0,1\}^\nat$,
    \item for every $f \oftype s_1 \times \ldots \times s_n \to s \in \asig$ 
          an \emph{interpretation}
          $\int{f} : A_{s_1} \times \ldots A_{s_n} \to A_s\,,$
    \item ${:} \in \asig$ with $\int{{\cons}}(x,\astr) = x \cons \astr$,
    \item\label{clause:01} $0,1\in\asig$ with $\int{0} = 0$ and $\int{1} = 1$.
  \end{enumerate} 
\end{definition}
The clause~\ref{clause:01} of the definition is optional; in fact, the results in this paper are independent of its presence.
We have included it since the models where $\int{0} = \int{1}$
are trivial, in the sense that then all bitstreams are equal.

\begin{definition}\normalfont
  Let $\aalg = \pair{A}{\sint}$ be a stream algebra.
  Moreover, let $\alpha : \avars \to A$ be a variable assignment.
  As usual, the \emph{interpretation of terms} $\sint^\aalg_{\alpha} : \TT \to A$ is defined inductively by:
  \begin{align*}
    \int{x}^\aalg_\alpha &= \alpha(x) &
    \int{f(t_1,\ldots,t_n)}^\aalg_\alpha &= \int{f}(\int{t_1}^\aalg_\alpha,\ldots,\int{t_n}^\aalg_\alpha)
  \end{align*}
  Then~$\aalg$ is called a \emph{(stream) model} of $\aes$
  if $\int{\ell}_\alpha = \int{r}_\alpha$ for every $\ell = r \in \aes$ and $\alpha : \avars \to A$.
  We write $\sint{}_\alpha$ for $\sint{}^\aalg_\alpha$ whenever $\aalg$ is clear from the context.
  For ground terms $t \in \ter{\asig}{\setemp}$, we have $\int{t}_\alpha = \int{t}_\beta$ 
  for all assignments $\alpha,\beta$;   we then write $\int{t}$ for short.
\end{definition}

Thus, we interpret function symbols as functions over bits and bitstreams as imposed by their sort.
In particular, terms of type $\sbitstream$ are interpreted as bitstreams.
In contrast to~\cite{rosu:2006}, our setup does not allow for \emph{confusion} in the models.
Recall that confusion means that the models can contain multiple representatives for the same stream.

\begin{definition}
  We say that a model $\aalg = \pair{A}{\sint}$ is \emph{full} if its domain contains all bitstreams, $A_{\sbitstream} = \{0,1\}^\nat$.
\end{definition}

\section{Turing Machines as Equational Specifications}\label{sec:turing}

We now define a set of standard equations (for bitstream specifications) that will be used throughout
 this paper:
\begin{gather}
  \hspace{-.1cm}
  \left.\begin{aligned}
    \zeros &= 0:\zeros \quad\;\; \ones = 1:\ones\\
    \zipn{1}{\tau} &= \tau \\
    \zipn{2}{x:\tau_1,\tau_2} &= x:\zipn{2}{\tau_2,\tau_1} \\
    \hspace{-.1cm}\zipn{n}{\tau_1,\sdots,\tau_n} &= \zipn{2}{\tau_1,\zipn{n-1}{\tau_2,\sdots,\tau_n}} &&
      \text{($n>2$)} \\
  \end{aligned}
  \hspace{-.3cm}\right\}\hspace{-.1cm}
  \label{eq:zap}
\end{gather}
%
To give an example, 
\short{
$
  \zipn{3}{\sigma,\tau,\rho} 
  = \sigma(0) : \tau(0) : \sigma(1) : \rho(0) : \sigma(2) : 
  \tau(1) : \sigma(3) : \rho(1) : \sigma(4) : \tau(2) 
  : \dots
$,
}  
\extended{
\begin{align*}
  \zipn{3}{\sigma,\tau,\rho} 
  =\ &\sigma(0) : \tau(0) : \sigma(1) : \rho(0) : \sigma(2) : \\
  &\tau(1) : \sigma(3) : \rho(1) : \sigma(4) : \tau(2) 
  : \dots  
\end{align*}}
writing $\sigma(i)$ for the $i$'th entry of the stream $\sigma$.

We emphasize that all systems of equations in this paper are finite.
To that end, we extend the specifications 
only by those equations from \eqref{eq:zap} that are needed by the specification,
that is, the equations $\zipn{n}{\ldots} = \ldots$ 
for which a symbol $\szipn{m}$ with $n\le m$ occurs in the specification.

\begin{lemma}\label{lem:zap}
  In every stream model $\aalg = \pair{A}{\sint}$ of a specification including 
  the equations from~\eqref{eq:zap}
  we have: 
  \begin{enumerate}
    \item $\int{\zeros} = 0^\omega$ \;and\;
          $\int{\ones} = 1^\omega$,
    \item for all $\astr_1,\ldots,\astr_k \in A_{\sbitstream}$, $k\ge2$ and $n\in\nat$:\\[.1ex]
          $\int{\szipn{1}}(\astr_1) = \astr_1$,\\[.1ex]
          $\int{\szipn{k}}(\astr_1,\ldots,\astr_k)(2n) = \astr_1(n)$\\[.1ex]
          $\int{\szipn{k}}(\astr_1,\ldots,\astr_k)(2n+1) = \int{\szipn{k-1}}(\astr_2,\ldots,\astr_k)(n)$ 
  \end{enumerate}
\end{lemma}

\noindent
A \emph{Turing machine} $\tm$ is a quadruple $\quadruple{\tmstates}{\tmsig}{\tmstart}{\stmtrans}$
consisting of
a finite set of states $\tmstates$,
an initial state $q_0 \in \tmstates$,
a finite alphabet $\tmsig$ containing a designated \emph{blank} symbol $\stmblank$,  and
a partial \emph{transition function} 
$\stmtrans \funin \autstates \times \tmsig \pto \tmstates \times \tmsig \times \{\tmL,\tmR\}$.
%

For convenience, we restrict $\tmsig$ to the alphabet $\tmsig = \{0,1\}$
where $0$ is the blank symbol $\stmblank$, and we 
denote Turing machines by triples $\triple{\tmstates}{\tmstart}{\stmtrans}$.
As input for the Turing machines we typically use a unary number representation $11\ldots1$ ($n$-times)
to encode the number $n$.
Of course, another encoding is possible, as long as the encoding is computable, and the Turing machine is able to detect the end of the input
(since $0$ is part of the input alphabet and it is also the blank symbol).
  
We define a translation of Turing machines to equational specifications of bitstream functions,
based on the standard translation to term rewriting systems from~\cite{tere:2003}.
However, we represent the tape using streams instead of finite lists, and
have one instead of four rules for `extending' the tape.
In particular, the equation for extending the tape is the equation for $\zeros$ from~\eqref{eq:zap}.
The terms of the shape $\bfunap{\astate}{\sigma}{\tau}$ represent configurations of the Turing machine, 
where the stream $\tau$ contains the tape content below and right of the head,
and $\sigma$ the tape content left of the head.
Notably, the head of the machine stands on the first symbol of $\tau$.
\begin{definition}\normalfont\label{def:tmtrs}
  Let $\tm = \triple{\tmstates}{\tmstart}{\stmtrans}$ be a Turing machine.
  We define the specification $\tmes{\tm}$ to consist of the following equations:
  \begin{align*}
    \bfunap{\astate}{x}{b:y} &= \bfunap{\astate'}{b':x}{y}
    &&\text{ for every }\tmtrans{\astate}{b} = \triple{\astate'}{b'}{\tmR}\\
    \bfunap{\astate}{a:x}{b:y} &= \bfunap{\astate'}{x}{a:b':y}
    &&\text{ for every }\tmtrans{\astate}{b} = \triple{\astate'}{b'}{\tmL}
  \end{align*}
  and for halting configurations additionally:
  \begin{align*}
    \bfunap{\astate}{x}{b:y} &= b 
    &&\text{ whenever }\tmtrans{\astate}{b}\text{ undefined}
  \end{align*}
  with the signature $\asig = \{0,1,{:}\} \cup \tmstates$
  with types
  $\astate \oftype \sbitstream \times \sbitstream \to \sbit$ for every symbol $\astate \in \tmstates$,
  and $0,1 \oftype \sbit$ and `${:}$' of type $\sbit \times \sbitstream \to \sbitstream$.
  Moreover, we use $\tmtrs{\tm}$ to denote the term rewriting system obtained from $\tmes{\tm}$
  by orienting all equations from left to right.
\end{definition}

Apart from the additional rule for termination, the translation $\tmtrs{\tm}$ is standard, 
and the rewrite rules model the transition relation of Turing machines
in one-to-one fashion.
So we take the liberty to define input of tuples $\tuple{n_1,\ldots,n_k} \in \nat^k$
and oracles directly on the term representations.
We pass $k$-tuples $\tuple{n_1,\ldots,n_k} \in \nat^k$ of natural numbers
as input to a Turing machine by choosing the following start configuration 
$\bfunap{\tmstart}{\zeros}{\;\zipn{k+1}{\num{k},\num{n_1},\;\ldots,\;\num{n_k}}}$
where $\num{n}$ stands for $(1:)^n\;\zeros$. 
The particular encoding of tuples is not crucial, 
but $\zipn{k+1}{\num{k},\num{n_1},\;\ldots,\;\num{n_k}}$ is 
for equational specifications more convenient than the G\"odel encoding.

We obtain machines with 
oracles $\oracle_1,\ldots,\oracle_m \subseteq \nat$
by writing the oracles elementwise interleaved on the tape left of the head:

\begin{notation}
  For $n \in \nat$ we use $\num{n}$ to abbreviate $(1:)^n : \zeros$.
  For $\oracle \subseteq \nat$, we let $\num{\oracle}$ denote the stream 
  $\chi_\oracle(0) : \chi_\oracle(1) : \chi_\oracle(2) : \ldots$
  where $\chi_\oracle$ is the characteristic function of $\oracle$.
%
%
  We write $\vec{\alpha}$ short for $\alpha_1,\ldots,\alpha_k$ 
  and $\num{\vec{\alpha}}$\, for $\num{\alpha_1},\ldots,\num{\alpha_k}$
  if $k$ is clear from the context.

  For a term rewriting system $\atrs$, we write $\to_\atrs$ for a rewrite step with respect to $\atrs$,
  and $\to^*_\atrs$ is the reflexive-transitive closure of $\to_\atrs$.
\end{notation}

\begin{definition}\label{def:init}
  Let $\tm = \triple{\tmstates}{\tmstart}{\stmtrans}$ be a Turing machine.
  Then for stream terms
  $\oracle_1,\ldots,\oracle_m \oftype \sbitstream$ and $n_1,\ldots,n_k \oftype \sbitstream$,
  we define 
  \begin{align*}
    \tminit{\oracle_1,&\ldots,\oracle_m}{n_1,\ldots,n_k} \defd\\
    &
    \bfunap{\tmstart}{\zipn{m}{\oracle_1,\ldots,\oracle_m}}{\;\zipn{k+1}{\num{k},n_1,\;\ldots,\;n_k}}
  \end{align*}
\end{definition}

\begin{definition}\normalfont
  A Turing machine $\tm = \triple{\tmstates}{\tmstart}{\stmtrans}$
  \emph{halts (with output $b$)} on inputs $n_1,\ldots,n_k \in\nat$ with oracles $\oracle_1,\ldots,\oracle_m \subseteq \nat$ 
  if there is a rewrite sequence
    $\tminit{\num{\oracle_1},\ldots,\num{\oracle_m}}{\num{n_1},\;\ldots,\;\num{n_k}}
    \to^*_{\tmtrs{\tm}} b$,
  where $b \in \{0,1\}$.
  Here $\num{\oracle}$ is short for the stream 
  $\chi_\oracle(0) : \chi_\oracle(1) : \chi_\oracle(2) : \ldots$
  where $\chi_\oracle$ is the characteristic function of $\oracle$.
\end{definition}
Note that the initial term is infinite due to the oracles, nevertheless
we consider only finite reduction sequences.
Due to the~rules for $\szipn{n}$ and $\zeros$,
there are infinite rewrite sequences even if the Turing machine halts.
However, $\tmtrs{\tm}$ is orthogonal and therefore
outermost-fair rewriting (or lazy evaluation) 
is normalizing, that is, computes the (unique) normal form $b \in \{0,1\}$ if it exists.

\begin{definition}\normalfont
  A \emph{$k$-ary predicate $P$ with $m$ oracles} is a relation $P \subseteq \powerset{\nat}^m \times \nat^k$.
  Then $P$ is called \emph{decidable} if there is
  a Turing machine $\tm$
  such that for all $\vec{\oracle} \in \powerset{\nat}^m$ and $\vec{n} \in\nat^k$:
  $\tm$ halts on input $\vec{n}$ with oracles $\vec{\oracle}$,
  and the output is $1$ if and only if $P(\vec{\oracle},\vec{n})$.
\end{definition}

In correspondence with Definition~\ref{def:init} we define
for $\oracle_1,\ldots,\oracle_m,$ $n_1,\ldots,n_k \in \{0,1\}^\omega$,
$\inttminit{\oracle_1,\ldots,\oracle_m}{n_1,\ldots,n_k}$ as shorthand for
$\bfunap{\int{\tmstart}}{\int{\szipn{m}}(\oracle_1,\ldots,\oracle_m)}{\;\int{\szipn{k+1}}(\int{\num{k}},n_1,\;\ldots,\;n_k)}$.
Then for the models of Turing machine specifications we have:
\begin{lemma}\label{lem:tm}
  Let $P \,{\subseteq}\, \powerset{\nat}^m {\times} \nat^k$ be decidable,
  and $\tm = \triple{\tmstates}{\tmstart}{\stmtrans}$ the corresponding Turing machine.
    Then in every stream model $\aalg = \pair{A}{\sint}$ of a specification including 
  the equations from~\eqref{eq:zap} and $\tmes{\tm}$
  we have for every $\vec{\oracle} \in \powerset{\nat}^m$ and $\vec{n} \in\nat^k$:
  $(\vec{\oracle},\vec{n}) \in P$ if and only if
  $\inttminit{\oracle_1,\ldots,\oracle_m}{n_1,\;\ldots,\;n_k} = 1$.
\end{lemma}

\begin{proof}
  $P$ is decidable, hence  
  $\tminit{\num{\oracle_1},\ldots,\num{\oracle_m}}{\num{n_1},\;\ldots,\;\num{n_k}}$
  has a nf in $\{0,1\}$, and the normal form is $1$ if and only if $(\vec{\oracle},\vec{n}) \in P$.
\end{proof}

\section{Levels of Undecidability}\label{sec:levels}

We briefly introduce complexity related notions that are relevant for this paper:
promise problems, reducibility, hardness and completeness, and the arithmetical and the analytical hierarchy.
For more details, 
we refer to the standard textbooks~\cite{shoe:1971,roge:1967}.

\begin{definition}\normalfont\label{def:membership}
  Let $A \subseteq P \subseteq \nat$.
  The \emph{promise (membership) problem for $A$ with promise $P$}
  is the question of deciding on the input of $n \in P$ whether $n \in A$.
  For the case $P = \nat$, we speak of the 
  \emph{membership problem for $A$}.
\end{definition}
We identify the membership problem for $A$ with the set $A$ itself,
and
the promise problem for $A$ with promise $P$ with the pair $\pair{A}{P}$,
also denoted by $\pp{A}{P}$.

\begin{definition}\normalfont\label{def:reduce}
  Let $A,B,P,Q \subseteq \nat$.
  Then \emph{$\pp{A}{P}$ can be (many-one) reduced to $\pp{B}{Q}$}, denoted $A \reducible B$,
  if there exists a partial recursive function $f \funin \nat \pto \nat$
  such that 
  $P \subseteq \domain{f}$, $f(P) \subseteq Q$, and
  $\myall{n \in P}{n \in A \Leftrightarrow \funap{f}{n} \in B}$.
\end{definition}

\begin{definition}\normalfont\label{def:hard}
  Let $B,Q \subseteq \nat$ and $\mathcal{P} \subseteq \powerset{\nat} \times \powerset{\nat}$.
  Then $\pp{B}{Q}$ is called $\mathcal{P}$-\emph{hard} 
  if every $\pp{A}{P} \in \mathcal{P}$ can be reduced to $\pp{B}{Q}$.
  Moreover, $\pp{B}{Q}$ is $\mathcal{P}$-\emph{complete} if additionally $\pp{B}{Q}$
  can be reduced to some $\pp{A}{P} \in \mathcal{P}$.
\end{definition}
We stress that Definition~\ref{def:hard} does not require that a $\mathcal{P}$-complete
promise problem $\pp{B}{Q}$ is member of $\mathcal{P}$ itself.
This allows for classifying promise problem using the
usual arithmetic and analytical hierarchy (for membership problems).

\begin{lemma}\label{lem:reduce}
  If $\pp{A}{P}$ can be reduced to $\pp{B}{Q}$
  and $\pp{A}{P}$ is $\mathcal{P}$-\emph{hard}, then $B$ is $\mathcal{P}$-\emph{hard}.
\end{lemma}

We use $\lencode{\cdot}$ to denote the well-known G\"odel encoding
of~finite \emph{lists} of numbers as elements of $\nat$: 
$\lencode{n_1,\sdots, n_k} := p_1^{n_1+1}\cdot \ldots \cdot p_k^{n_k+1}$, 
where $p_1 < p_2 < \ldots < p_k$ are the first $k$ prime numbers.

We define the arithmetical and analytical hierarchies:
\begin{definition}\normalfont\label{def:arithclasses}
  Let $\csig{0}{0}:= \cpi{0}{0} := \cdel{0}{0}$ be the collection of 
  recursive sets of natural numbers
  (the decidable problems).
  Then for $n \geq 1$, we define:
  \begin{itemize}
    \item 
      $\csig{0}{n}$ consists of sets $\{n \,{\mid}\, \myex{x \,{\in}\, \nat}{\lencode{x,n} \,{\in}\, B}\}$ 
      with~$B \in \cpi{0}{n-1}$,
    \item 
      $\cpi{0}{n}$ consists of sets $\{n \,{\mid}\, \myall{x \,{\in}\, \nat}{\lencode{x,n} \,{\in}\, B}\}$ 
      with~$B \in \csig{0}{n-1}$,
    \item 
      $\cdel{0}{n} := \csig{0}{n} \cap \cpi{0}{n}$.
  \end{itemize}
  The \emph{arithmetical hierarchy} consists of the classes
  $\cpi{0}{n}$, $\csig{0}{n}$ and $\cdel{0}{n}$ for $n \in \nat$.
\end{definition}

For example, the membership $a \in A$ for every set $A \in \cpi{0}{2}$ can be defined
by a formula of the form 
$\myall{x_1}{\myex{x_2}{\myall{x_3}{P(a,x_1,x_2,x_3)}}}$
where $P$ is a decidable predicate.

The analytical hierarchy extends this classification of sets to formulas of the language of second-order arithmetic,
\pagebreak
that is, with set (or equivalently function) quantifiers.
The following definition makes use of a result from recursion theory, see~\cite{roge:1967}, stating that
if there is at least one set quantifier, then two number quantifiers suffice
(for functions quantifiers, one number quantifier suffices).

\begin{definition}\normalfont\label{def:analytic}
  Let $\csig{1}{0}:= \cpi{1}{0} := \cdel{1}{0} = \bigcup_{n\in\nat} \cpi{0}{n}$ be the set of all arithmetic predicates.
  A set $A \subseteq \nat$ is in $\cpi{1}{n}$ for $n>0$ if 
  there is a decidable predicate~$P$ with $m$ oracles
  such that for all $a \in \nat$:
  \begin{align*}
    a \in A &\iff \myall{\oracle_1}{\myex{\oracle_2}{{\sdots} \myex{\oracle_m}{\myall{x_1}{\myex{x_2}{P(\oracle_1,\sdots,\oracle_n,a,x_1,x_2)}}}}}\\
    a \in A &\iff \myall{\oracle_1}{\myex{\oracle_2}{{\sdots} \myall{\oracle_m}{\myex{x_1}{\myall{x_2}{P(\oracle_1,\sdots,\oracle_n,a,x_1,x_2)}}}}}
  \end{align*}
  for $n$ even, and $n$ odd, respectively.
  Here, $\oracle_1,\ldots,\oracle_m \subseteq \nat$, the corresponding quantifiers are set quantifiers,
  and $x_1,x_2 \in \nat$ with number quantifiers.
  Then $A$ is in $\csig{1}{n}$, if the condition holds with all $\forall$ and $\exists$ quantifiers swapped.
  Finally, $\cdel{1}{n} = \cpi{1}{n} \cap \csig{1}{n}$.
\end{definition}

\section{Equality in Models}\label{sec:models:theoretic}


In this section we study the complexity of different model-theoretic semantics
of equivalence of bitstream specifications. 
Based on the notion of models for bitstream specifications from Section~\ref{sec:specs},
we first formalize the equivalences that we consider.

For all of the following model-theoretic equivalences, 
we have the choice whether or not we require the models
to be full, that is, their domain contains all bitstreams.
For example, we can consider the equality of terms in all models or in all full models:
\begin{definition}
  Let $\aes$ be a bitstream specification over $\asig$, and $s,t \in \TT$ with $s,t \oftype \sbitstream$.
  Then $s$ and $t$ are said to be
  \begin{itemize}
    \item 
      \emph{equal in all models of $\aes$} if
      \begin{center}\vspace{-.5ex}
      $\aalg \models \aes$ implies $\aalg \models s = t$
      for all stream algebras $\aalg$\,,
      \end{center}
    \item 
      \emph{equal in all full models of $\aes$} if
      \begin{center}\vspace{-.5ex}
      $\aalg \models \aes$ implies $\aalg \models s = t$
      for all full stream algebras $\aalg$\,.
      \end{center}
  \end{itemize}
\end{definition}

The set of solutions of a term $s$ in a specification $\aes$ 
is the set of interpretations $\int{s}$ of $s$ in all models satisfying $\aes$:
\begin{definition}
  Let $\aes$ be a bitstream specification over $\asig$, and $s \in \TZ$ with $s \oftype \sbitstream$.
  Then the set of
  \begin{itemize}
    \item 
      \emph{solutions of $s$ in $\aes$ with respect to all models} is\vspace{-.5ex}
      \begin{align*}
        \sol{s}{\aes} = \{\,\int{s}^{\aalg} \mid \aalg \models \aes \,\}\,,
      \end{align*}
    \item 
      \emph{solutions of $s$ in $\aes$ with respect to all full models} is\vspace{-.5ex}
      \begin{align*}
        \solful{s}{\aes} = \{\,\int{s}^{\aalg} \mid \text{$\aalg$ full}, \aalg \models \aes \,\}\,.
      \end{align*}
  \end{itemize}
\end{definition}
Here it suffices to consider only ground terms $s \in \TZ$. 
For terms $t \in \TT$ with variables, the set of solutions can be defined
as $\sol{t}{\aes} = \{\,\int{t}^{\aalg}_\alpha \mid \aalg \models \aes, \alpha : \avars \to A \,\}$.
However, then $\sol{t}{\aes} = \sol{s}{\aes}$ if $s$ is the ground term obtained from~$t$ by
interpreting the variables in $t$ as fresh constants (formally, this amounts to an extension of the signature).

\begin{definition}
  Let $\aes_s$ and $\aes_t$ be bitstream specifications over $\asig_s$ and $\asig_t$, respectively.
  Let $s \in \ter{\asig_s}{\setemp}$ and $t \in \ter{\asig_t}{\setemp}$.
  Then $s$ and $t$ have
  \begin{itemize}
    \item \emph{equal solutions over all models} if $\sol{s}{\aes_s} = \sol{t}{\aes_t}$,
    \item \emph{equal solutions over all full models} if $\solful{s}{\aes_s} = \solful{t}{\aes_t}$.
  \end{itemize}
\end{definition}

\begin{definition}
  Let $\aes$ be bitstream specifications over $\asig$, and $s \in \ter{\asig}{\setemp}$.
  Then $s$ is said to have
  \begin{itemize}
    \item \emph{a unique solution over all models} if $\setsize{\sol{s}{\aes}} = 1$,
    \item \emph{a unique solution over all full models} if $\setsize{\solful{s}{\aes}} = 1$,
    \item \emph{a solution over all models} if $\setsize{\sol{s}{\aes}} \ge 1$,
    \item \emph{a solution over all full models} if $\setsize{\solful{s}{\aes}} \ge 1$,
    \item \emph{at most one solution over all models} if $\setsize{\sol{s}{\aes}} \le 1$,
    \item \emph{at most one solution over all full models} if $\setsize{\solful{s}{\aes}} \le 1$.
  \end{itemize}
\end{definition}

\subsection{Auxiliary Definitions}

First, we define a few (systems of) equations that are repeatedly used throughout this section.
The following function $\siszeros$ that maps $\zeros$ to $\ones$ and every other bitstreams to $\zeros$:
\begin{gather}
  \left.
  \begin{aligned}
  \iszeros{\zeros} &= \ones &
  \iszeros{0:\astr} &= \iszeros{\astr} \\
  && \iszeros{1:\astr} &= \zeros
  \end{aligned}
  \hspace{0.3cm}\right\}\hspace{-.3cm}
  \label{eq:filter}
\end{gather}
This function does exactly what its name suggests; 
it checks whether the argument is the stream of zeros.
We use the bit $0$ or the stream $\zeros$ for \emph{false},
and $1$ and $\ones$ for \emph{true}.

We focus on specifications of bitstreams, and 
encode streams of natural numbers as bitstreams
via the sequence of run\nb-length of ones.
For instance, the stream $3:1:0:2:\ldots$ is encoded as
$1:1:1:0:1:0:0:1:1:0:\ldots$.
We then define functions $\suhd$ and $\sutl$ that 
are the unary counterpart for head and tail on streams of natural numbers:
\begin{gather}
  \left.
  \begin{aligned}
  \uhd{0:\astr} &= \zeros & \utl{0:\astr} &= \astr \\
  \uhd{1:\astr} &= 1:\uhd{\astr} & \utl{1:\astr} &= \utl{\astr}
  \end{aligned}
  \hspace{0.25cm}\right\}\hspace{0cm}
  \label{eq:bhd}
\end{gather}
For instance, we have 
\begin{align*}
  \uhd{1:1:1:0:1:0:0:1:1:\ldots} &= 1:1:1:\zeros \\
  \utl{1:1:1:0:1:0:0:1:1:\ldots} &= 1:0:0:1:1:\ldots
\end{align*}
The following lemma summarizes these properties:
\begin{lemma}\label{lem:iszeros}
  In every stream model $\aalg = \pair{A}{\sint}$ of a specification including 
  the equations from~\eqref{eq:filter} and \eqref{eq:bhd}
  we have: 
  \begin{enumerate}
    \item $\int{\siszeros}(0^\omega) = 1^\omega$,\\[.1ex]
          $\int{\siszeros}(w) = 0^\omega$ for every $w \in A_{\sbitstream} \setminus \{0^\omega\}$,
    \item $\int{\suhd}(1^n\,0\,w) = 1^n\,0^\omega$ for every $w \in A_{\sbitstream}$,\\[.1ex]
          $\int{\suhd}(1^\omega) = 1^\omega$,
    \item $\int{\sutl}(1^n\,0\,w) = w$ for every $w \in A_{\sbitstream}$. 
  \end{enumerate}
\end{lemma}
Note that all interpretations are uniquely defined,
apart from the combination $\int{\sutl}(1^\omega)$ which can be any stream depending on the model.
To avoid this case, we need means to ensure that a certain bitstream is a valid encoding
of a stream of natural numbers, that is, the stream contains infinitely many zeros:
\begin{gather}
  \left.
  \begin{aligned}
   \natstr{\ones} &= \zeros &
   \natstr{0:\astr} &= 1:\natstr{\astr} \\
   &&\natstr{1:\astr} &= \natstr{\astr} \\
  \end{aligned} 
  \hspace{.1cm}\right\}\hspace{-0.1cm}
  \label{eq:natstr}
\end{gather}
Then an equation 
$\natstr{X} = \ones$ guarantees that $\int{X}$ 
represents a stream of natural numbers:

\begin{lemma}\label{lem:natstr}
  In every stream model $\aalg = \pair{A}{\sint}$ of a specification including 
  the equations from~\eqref{eq:natstr}
  we have: $\int{\snatstr}(w) = 1^\omega$
  if and only if $w$ contains infinitely many zeros.
\end{lemma}

\begin{proof}
  The equations on the right `walk' over the stream,
  deleting $1$'s and converting $0$'s to $1$'s.
  If the stream contains infinitely many $0$'s,
  then an infinite stream of $1$'s will be produced.
  However, if some tail of the stream contains only $1$'s then the equation on the left
  ensures that the interpretation is unequal to $1^\omega$.
\end{proof}

\begin{definition}\label{def:canonical}
  Let $\tm = \triple{\tmstates}{\tmstart}{\stmtrans}$ be a Turing machine.
  Then the \emph{canonical model $\aalg = \pair{A}{\sint}$} for the union
  of the specifications $\tmes{\tm}$,
  \eqref{eq:zap}, \eqref{eq:filter}, \eqref{eq:bhd} and \eqref{eq:natstr}
  consists of the domain $A_{\sbitstream} = \{0,1\}^\nat$
  with interpretations $\sint$ 
  as given in Lemmas~\ref{lem:iszeros} and~\ref{lem:natstr},
  extended by
  \begin{enumerate}
    \item $\int{\sutl}(1^\omega) = 1^\omega$, 
    \item for every $\oracle_1,\ldots,\oracle_m,n_1,\ldots,n_k \subseteq \nat$:\\
          $\int{q}(\vec{\oracle},\vec{n}) = 1$ whenever
          $\int{q}(\num{\vec{\oracle}},\num{\vec{n}}) \to^* 1$,
          and\\
          $\int{q}(\vec{\oracle},\vec{n}) = 0$ otherwise.

  \end{enumerate}
\end{definition}

\begin{lemma}\label{lem:canonical}
  The canonical model is a model of the union of the equational specifications 
   $\tmtrs{\tm}$, \eqref{eq:zap}, \eqref{eq:filter}, \eqref{eq:bhd} and \eqref{eq:natstr}.
\end{lemma}

\begin{proof}
  The rewrite system $\tmtrs{\tm}$ is orthogonal, consequently we have finitary confluence 
  and infinitary unique normal forms~\cite{tere:2003}.
  Hence, we can employ a normal forms semantics for $\int{q}$
  (where we map terms without normal forms to $0$).
  For the remaining equations, it is easy to see that the chosen semantics forms a model.
\end{proof}

\subsection{Equality in all Models}\label{sec:models}

For the complexity of equality in all models we obtain:

\begin{theorem}\label{thm:models}
  The following problem is $\cpi{1}{1}$\nb-complete:
  \problem{Bitstream specification $\aes$, terms $s,t \oftype \sbitstream$.}{Are $s$ and $t$ equal in all models of $E$?}
  \noindent
\end{theorem}
\begin{proof}
  The well-foundedness problem for decidable binary relations is known to be $\cpi{1}{1}$-complete,
  that is, the problem of deciding on the input of a decidable binary predicate $M \subseteq \nat \times \nat$ 
  (given in the form of a Turing machine), whether $M$ is well-founded.
  We reduce this problem to an equality problem.
  Let $M \subseteq \nat \times \nat$ be a decidable predicate,
  and $\tm = \triple{\tmstates}{\tmstart}{\stmtrans}$ the corresponding Turing machine.
  We define the following specification $\aes$:
  \begin{align*}
    &\sstart = \iszeros{\srun(1,\msf{X})} \quad\quad \natstr{X} = \ones\\
    &\srun(0,\astr) = \ones \\[-3.3ex]
    &\srun(1,\astr) = 0:\srun(\overbrace{\tminit{\zeros}{\uhd{\astr},\uhd{\utl{\astr}}}}^{\Phi(\astr)},\utl{\astr})
  \end{align*}
  together with the equations from $\tmes{\tm}$ and~\eqref{eq:zap}, \eqref{eq:filter}, \eqref{eq:bhd} and \eqref{eq:natstr}.
  We prove that: $\aes \models \sstart = \zeros$ if and only if $M$ is well-founded.
  
  For `$\Rightarrow$' let $M$ be non-well-founded,
  and $n_0 \mathrel{M} n_1 \mathrel{M} n_2 \mathrel{M} \ldots$ be an infinite chain.
  We construct a $\asig$-algebra $\aalg = \pair{A}{\sint}$ such that $\aalg \models \aes$ 
  but not $\aalg \models \sstart = \zeros$.
  We define $\aalg$ as an extension of the canonical model (Definition~\ref{def:canonical}).
  The values of $\int{\Phi(\astr)}$ and $\int{\utl{\astr}}$ are determined by the canonical model,
  and together with the equations for $\srun$
  we obtain for every stream $\oracle \in \{0,1\}^\omega$:
  $\int{\srun}(0,\oracle) = 1^\omega$,
  and
  $\int{\srun}(1,\oracle) = 0:\int{\srun}(\int{\Phi(\num{\oracle})},\int{\sutl}(\oracle))$.
  Hence, there is a unique interpretation $\int{\srun}$ 
  that results in a model for the equations of $\srun$.
  We define $\kappa_i = 1^{n_{i}} \, 0 \, 1^{n_{i+1}} \, 0 \, 1^{n_{i+2}} \ldots$.
  and we let $\num{n} = 1^n\,0^\omega$. Then for $i \in \nat$ we have 
  \begin{align*}
    \int{&\srun}(1,\kappa_i) 
      = 0:\int{\srun}(\int{\Phi}(\kappa_i),\kappa_{i+1})\\
      &= 0:\int{\srun}(\int{\tm}(\num{n_i},\num{n_{i+1}}),\kappa_{i+1})
      = 0:\int{\srun}(1,\kappa_{i+1})
  \end{align*}
  since we have that
  $\int{\suhd}(\kappa_j) = \num{n_j}$ and 
  $\int{\sutl}(\kappa_j) = \kappa_{j+1}$ for all $j\in\nat$
  by Lemma~\ref{lem:iszeros}.
  Thus, $\int{\srun}(1,\kappa_0) = 0^\omega$.
  Let $\int{\msf{X}} = \kappa_0$ and $\int{\sstart} = 1^\omega$.
  Then $\int{\snatstr}(\int{\msf{X}}) = \int{\ones}$ by Lemma~\ref{lem:natstr},
  and $\int{\sstart} = \int{\siszeros}(\int{\srun}(1,\int{\msf{X}}))$ by Lemma~\ref{lem:iszeros}.
  We have constructed a model, where $\int{\sstart} = 1^\omega$,
  and, hence, $\aes \not\models \sstart = \zeros$.

  For `$\Leftarrow$' let $M$ be well-founded. Let $\aalg$ be a $\asig$-algebra
  such that $\aalg \models E$. We show that $\int{\sstart} = 0^\omega$.
  Since $\int{\snatstr}(\int{\msf{X}}) = \int{\ones}$, 
  $\int{\msf{X}}$ contains infinitely many zeros by Lemma~\ref{lem:natstr}.
  Thus, $\int{\msf{X}} = 1^{n_0} \,0\, 1^{n_1} \,0\, 1^{n_2} \ldots$
  for some $n_0$, $n_1,n_2,\ldots \in \nat$.
  Let $\kappa_i = 1^{n_{i}} \,0\, 1^{n_{i+1}} \,0\, 1^{n_{i+2}}\ldots$ for $i \in \nat$.
  Then
  \begin{align*}
    \int{\srun}(1,\kappa_i) &= 0:\int{\srun}(\int{\tm}(\num{n_i},\num{n_{i+1}}),\kappa_{i+1})\\
    & = \begin{cases}
          \int{\srun}(1,\kappa_{i+1}) &\text{if $\int{\tm}(\num{n_i},\num{n_{i+1}}) = 1$}\\
          \int{\srun}(0,\kappa_{i+1}) = 1^\omega &\text{if $\int{\tm}(\num{n_i},\num{n_{i+1}}) = 0$}
        \end{cases}
  \end{align*}
  Hence, $\int{\srun}(1,\int{\msf{X}}) = 0^\omega$ if and only if
  $\int{\tm}(\num{n_i},\num{n_{i+1}}) = 1$ for all $i\in\nat$.
  However, this would contradict well-foundedness of $M$.
  As a consequence, we obtain that $\int{\srun}(1,\int{\msf{X}}) \ne 0^\omega$
  and $\int{\sstart} = \iszeros{\int{\srun}(1,\int{\msf{X}})} = 0^\omega$ by Lemma~\ref{lem:iszeros}.
  This concludes the $\cpi{1}{1}$-hardness proof.
  
  To show $\cpi{1}{1}$-membership, we resort to the L\"owenheim--Skolem theorem. It 
  states that if a formula of first-order predicate logic has an uncountable model, 
  then it also has a countable model.
\pagebreak
  Here, we employ that the domain $A_{\sbitstream}$ can be encoded as an arbitrary set
  with functions $\int{\shd} \oftype A_{\sbitstream} \to \{0,1\}$ and 
  $\int{\stl} \oftype A_{\sbitstream} \to A_{\sbitstream}$
  together with a first-order predicate logic formula that 
  excludes confusion, that is, 
  elements $a,b \in A_{\sbitstream}$ with $\int{\shd}{\int{\stl}^n}(a) = \int{\shd}{\int{\stl}^n}(b)$ for all $n\in\nat$
  are required to be equal, that is, $a = b$.
  Likewise, the interpretations of the symbols in $\asig$
  can be translated to first-order predicates,
  and validity of the equations to first-order formulas. 
  As a consequence, 
  $\aalg \models \aes \wedge \int{s} \ne \int{t}$
  can be expressed as first-order formula,
  and if it has a model, then also a countable one.
  Hence, it suffices in
  $\myall{\aalg}{\aalg \models \aes \implies \int{s} = \int{t}}$
  to quantify over countable models.
  For this purpose of quantifying over countable models, 
  a set quantifier $\forall \aalg \subseteq \nat$ suffices.
  This proves $\cpi{1}{1}$-membership.
\end{proof}

The following three results are obtained by slight adaptations of the proof of Theorem~\ref{thm:models}.
\extended{
In the proof of Theorem~\ref{thm:models}, we have $\aes \models \sstart = \zeros$
if and only if $\sstart$ has a unique solution over all models of $\aes$.
As a consequence,
we obtain the following results concerning (unique) solvability:
}

\short{
\begin{theorem}\label{thm:atmost}\label{thm:atleast}\label{thm:unique}
  The following problems:%
  \problem{Bitstream specification $\aes$, ground term $s \oftype \sbitstream$.}{%
    Does $s$ have 
    (i)~at most one solution,
    (ii)~a solution, and
    (iii)~a unique solution
    over all models of $E$?}%
  \noindent
  are (i)~$\cpi{1}{1}$-complete, 
  (ii)~$\csig{1}{1}$-complete, and 
  (iii)~$\cpi{1}{1}$\nb-hard, $\csig{1}{1}$-hard and strictly contained in $\cdel{1}{2}$.
\end{theorem}
}

\extended{
\begin{theorem}\label{thm:atmost}
  The following problem is $\cpi{1}{1}$\nb-complete:
  \problem{Bitstream specification $\aes$, term $s$.}{Does $s$ have at most one solution over all models of $E$?}
\end{theorem}

\extended{
\begin{proof}
  The $\cpi{1}{1}$-hardness follows from the proof of Theorem~\ref{thm:models},
  as $\sstart$ has $\le 1$ solutions if and only if $M$ is well-founded.

  The membership in $\cpi{1}{1}$ uses that it suffices to consider countable models
  as in the proof of Theorem~\ref{thm:models}.
  Then the formula characterizes the property having at most one solution:
    $\myall{\aalg_1}{\myall{\aalg_2}{(\aalg_1 \models E) \wedge (\aalg_2 \models E) \implies \int{s}^{\aalg_1} = \int{s}^{\aalg_2}}}$.
  The two $\forall$ set quantifiers can be merged into one,
  and the properties $\aalg \models E$, and $\int{s}^{\aalg_1} = \int{s}^{\aalg_2}$ are arithmetic.
  Hence, the property is in $\cpi{1}{1}$.
\end{proof}
}

\begin{theorem}\label{thm:atleast}
  The following problem  is $\csig{1}{1}$\nb-complete:
  \problem{A bitstream specification $\aes$, a term $s$.}{Has $s$ a solution over all models of $E$?}
\end{theorem}

\extended{
\begin{proof}
  The $\csig{1}{1}$-hardness follows from a tiny adaptation of the proof of Theorem~\ref{thm:models}.
  We replace the equation $\sstart = \iszeros{\srun(1,\msf{X})}$ by
  the equations $\sstart = \srun(1,\msf{X})$ and $\sstart = \zeros$.
  Then every model where $\int{\srun(1,\msf{X})} \ne 0^\omega$ is ruled out,
  and hence, the specification has a model, and $\sstart$ a solution,
  if and only if $M$ is not well-founded.
  
  The membership in $\csig{1}{1}$ can be described by the following formula
  (we again use that we only need to quantify over countable $\asig$-algebras):
  $\myex{\aalg}{\aalg \models E}$.
  Hence, the property is in $\csig{1}{1}$.
\end{proof}
}

\begin{theorem}\label{thm:unique}
  The following problem  is $\cpi{1}{1}$\nb-hard, $\csig{1}{1}$-hard and strictly contained in $\cdel{1}{2}$:
  \problem{A bitstream specification $\aes$, a term $s$.}{Has $s$ unique solutions over all models of $E$?}
\end{theorem}

\extended{
\begin{proof}
  The $\cpi{1}{1}$-hardness follows from the fact that
  the specification used in the proof of Theorem~\ref{thm:atmost}
  always has a solution; then unique solvability coincides with at most one solution.
  
  The $\csig{1}{1}$-hardness is a consequence of the fact that
  the specification used in the proof of Theorem~\ref{thm:atleast}
  always has at most one solution (due to the equation $\sstart = \zeros$);
  then unique solvability coincides with at least one solution.

  For the $\cdel{1}{2}$-membership we observe that
  a term $s$ has a unique solution if and only if $s$ has at least and $s$ has 
  at most one solution.
  Therefore unique solvability can be described by the conjunction of a $\cpi{1}{1}$- and a $\csig{1}{1}$-formula.
\end{proof}
}
}

\subsection{Equality in all Full Models}\label{sec:models:full}

In Section~\ref{sec:models} we have considered models whose domain 
was any non-empty set of bitstreams ($A_{\sbitstream} \subseteq \{0,1\}^\omega$).
However, when writing equations such as
$\even{x:y:\tau} = x:\even{\tau}$,
the intended semantics is often that these equations should hold
for all streams, that is, in full models with domain $A_{\sbitstream} = \{0,1\}^\omega$.
We find that the restriction to full models
results in a huge jump of the complexity, which then subsumes the entire analytical hierarchy.

\newcommand{\sisnat}{\msf{nat}}
\newcommand{\isnat}{\funap{\sisnat}}
To prepare for the proof, we introduce some auxiliary specifications.
We define $\sisnat$
such that an equation $\isnat{X} = \ones$ guarantees that 
the interpretation $\int{X}$ represents a natural number in unary encoding,
that is, $\int{X} = 1^n \,0^\omega$ for $n\in\nat$, as follows:
\begin{gather}
  \left.
  \begin{aligned}
  \isnat{0:1:\sigma} &= \zeros & \isnat{1:\sigma} &= \isnat{\sigma}\\
  \isnat{0:0:\sigma} &= \isnat{0:\sigma} &
  \isnat{\ones} &= \zeros
  \end{aligned}
  \hspace{.25cm}\right\}\hspace{-0.25cm}
  \label{eq:isnat}
\end{gather}

\begin{lemma}\label{lem:isnat}
  In every stream model $\aalg = \pair{A}{\sint}$ of a specification including 
  the equations from~\eqref{eq:isnat}
  we have: if $\int{\sisnat}(w) = 1^\omega$
  then $w = 1^n \,0^\omega$ for some $n\in\nat$.
\end{lemma}

\begin{proof}
  If a stream is not of the format $1^n\,0^\omega$ for some $n\in\nat$
  then it is $1^\omega$ or contains $\ldots 01 \ldots$.
  The last equation rules out the case $1^\omega$ 
  (ensures that the interpretation is not $1^\omega$).
  
  The first three equations are exhaustive in the sense that 
  every stream can be matched by one of them.
  The first equation rules out 
  streams that contain a $1$ after a $0$,
  and the equations two and three `walk' step by step over the stream (proceed with the tail).
\end{proof}

We moreover define a function $\slesseq$ such that $\lesseq{X}{Y} = \ones$
guarantees that pointwise $\int{X} \le \int{Y}$:
\begin{gather}
  \left.
  \begin{aligned}
  \lesseq{0:\astr}{x:\bstr} &= \lesseq{\astr}{\bstr}\\
  \lesseq{1:\astr}{1:\bstr} &= \lesseq{\astr}{\bstr}\\
  \lesseq{1:\astr}{0:\bstr} &= \zeros
  \end{aligned}
  \hspace{.75cm}\right\}\hspace{-0.75cm}
  \label{eq:lesseq}
\end{gather}

\begin{lemma}\label{lem:lesseq}
  In every stream model $\aalg = \pair{A}{\sint}$ of a specification including 
  the equations from~\eqref{eq:lesseq}
  we have that if $\int{\slesseq}(\astr,\bstr) = 1^\omega$,
  then $\astr$ is pointwise $\le$ than $\bstr$ (for all $\astr,\bstr \in A_{\sbitstream}$).
\end{lemma}

Lemmas~\ref{lem:isnat} and~\ref{lem:lesseq} are valid for non-full models as well.
As explained in the introduction, the assumption of full models is
crucial to guarantee that equations with variables have to hold
for all streams (assigned to the variables) and not only the streams in the model.

\begin{theorem}\label{thm:models:full}
  The following problem  subsumes the analytical hierarchy:
  \problem{Bitstream specification $\aes$, terms $s,t \oftype \sbitstream$.}{Are $s$ and $t$ equal in all \emph{full} models of $E$?}
\end{theorem}

The idea of the proof is as follows.
We translate formulas of the analytical hierarchy into stream specifications
by representing $\forall$ set quantifiers by equations with variables.
This simulates a quantification over all streams as the models are \emph{full},
and the equations have to hold for all assignments of the variables.

The $\exists$ set quantifiers are eliminated in favor of Skolem functions~$f$, 
that is, axioms of the form 
$\myall{\vec{x}}{\myex{y}{\psi(x_1,\ldots,x_n,y)}}$
are replaced by 
$\myall{\vec{x}}{\psi(x_1,x_2,\ldots,x_n,f(x_1,\ldots,x_n))}$.
%
The interpretation of these functions is determined by the model,
and the question whether there exists a model corresponds 
to an existential quantification over all Skolem functions.

\begin{proof}
  For every analytical set $A$, we reduce the membership problem in $A$
  to~an equality problem. 
  Every set $A$ of the analytical hierarchy can be defined by
  \begin{align}\label{eq:proof:full}
    a \not\in A \;\;  \Longleftrightarrow\;\;& \\
    \myall{\oracle_1}{\myex{\oracle_2}{&\myall{\oracle_3}{\ldots\myex{\oracle_n}{\;\myall{x_1}{\myex{x_2}{\;
      M(\oracle_1,\ldots,\oracle_n,a,x_1,x_2)
    }}}}}} \nonumber
  \end{align}
  where $n \in \nat$ is even (without loss of generality since $\cpi{1}{n} \subset \cpi{1}{n+1}$) 
  and $M$ a decidable predicate.
  Let $\tm = \triple{\tmstates}{\tmstart}{\stmtrans}$ the Turing machine corresponding to $M$.
  Let $a \in \nat$ be given.
  We define $\aes$ to be the following system of equations:
  \begin{align*}
    &\sstart(\tau_1,\tau_3,\ldots,\tau_{n-1}) = \srun(1,\: \zipn{n}{\tau_1,\msf{g_2}(\tau_1),\tau_3,\msf{g_4}(\tau_1,\tau_3),\\
      &\hspace{3cm}\ldots,\tau_{n-1},\msf{g_n}(\tau_1,\tau_3,\ldots,\tau_{n-1})},\: \zeros) \\ 
    &\sstart(\tau_1,\tau_3,\ldots,\tau_{n-1}) = \zeros\\
    &\srun(0,\tau,\gamma_1) = \ones \\
    &\srun(1,\tau,\gamma_1) = 0:\srun(\tminit{\tau}{A,\gamma_1,\msf{h_2}(\tau,\gamma_1)},\;\tau,\; 1:\gamma_1)\\
    &A = (1:)^a\;\zeros\\
    &\isnat{\msf{h_2}(\tau,\gamma_1)} = \ones
  \end{align*}
  together with the equations from $\tmes{\tm}$, \eqref{eq:zap}, and \eqref{eq:isnat}.
  The symbols $g_{2i}$ are typed $\sbitstream^i \to \sbitstream$.
  We claim:  $\aes \models \zeros = \ones$ if and only if $a \in A$.
  For this purpose it suffices to show that the specification has a model ($\myex{\aalg}{\aalg \models \aes}$) 
  if and only if the formula in the right-hand side of~\eqref{eq:proof:full} is valid.  
  
  The idea is that the specification models a Skolem normal form of the 
  analytical formula in~\eqref{eq:proof:full}.
  The $\forall$ set quantifiers are modeled by an equation with stream variables;
  recall that equations have to hold for all assignments of the variables.
  In particular, the variables $\tau_{1},\tau_3,\ldots,\tau_{n-1}$ in the 
  first equation $\sstart(\tau_1,\tau_3,\ldots,\tau_{n-1}) = \ldots$
  model the set quantifiers $\forall \oracle_1,\ldots,\forall \oracle_{n-1}$, respectively. 
  The $\exists$ set quantifiers are modeled by Skolem functions
  $g_2,g_4,\ldots,g_n$ which in the specification are stream functions
  that get the value of the preceding $\forall$ quantifiers as arguments.
  These stream functions~$g_{2i}$ are unspecified and can be `freely chosen' by the model~$\aalg$.
  Thus, the existential quantification over the Skolem functions
  corresponds to the existential quantification over all models in $\myex{\aalg}{\aalg \models \aes}$.

  The streams $\tau_1,g_2(\tau_i),\ldots,\tau_{n-1},g_n(\tau_1,\tau_3,\ldots,\tau_{n-1})$ 
  that represent the values of the set quantifiers are then
  interleaved by $\szipn{n}$, and passed as the second argument, named $\tau$, to $\srun$;
  this argument serves as the left side of the tape for every invocation of the Turing machine $\tm$.

  The $\forall x_1$ number quantifier is modeled by the third argument $\gamma_1$ of $\srun$.
  The initial value of $\gamma_1$ is $\zeros$, and `$1:\cxthole$' is prepended
  (corresponding to counting up)
  each time the Turing machine halts with output $1$.
  The number quantifier $\exists x_2$ is modeled by the Skolem function $\msf{h}_2$
  for which the equation $\isnat{\msf{h_2}(\tau,\gamma_1)} = \ones$
  ensures by Lemma~\ref{lem:iszeros} that the interpretation $\int{\msf{h_2}(\tau,\gamma_1)}$
  is a unary encoding of a natural number.
  Then the term $\tminit{\tau}{A,\gamma_1,\msf{h_2}(\tau,\gamma_1)}$
  with $\tau = \zipn{n}{\tau_1,\msf{g_2}(\tau_1),\tau_3,\msf{g_4}(\tau_1,\tau_3),\ldots,\tau_{n-1},\msf{g_n}(\tau_1,\tau_3,\ldots,\tau_{n-1})}$
  corresponds precisely to $M(\oracle_1,\ldots,\oracle_n,a,x_1,x_2)$ in~\eqref{eq:proof:full}.

  For `$\Leftarrow$', assume that the formula in~\eqref{eq:proof:full} is valid.
  We construct a model $\aalg = \pair{A}{\sint}$ as an extension of the canonical model 
  (Definition~\ref{def:canonical}).
  For $\int{g_2},\int{g_4},\ldots,\int{g_n},\int{h_2}$ we pick the Skolem functions
  for the quantifiers $\exists \oracle_2,\exists \oracle_4,\ldots,\exists \oracle_n,\exists x_2$, respectively
  (where $\int{h_2}$ is a stream function that works on the unary encoding of natural numbers).
  For $\sigma \in \{0,1\}^\omega$, we define 
  $\int{\sisnat}(\sigma) = 1^\omega$ if $\sigma$ is of the form $1^n \,0^\omega$,
  and $0^\omega$, otherwise.
  The definition of $\int{\srun}$ is analogous to the proof of Theorem~\ref{thm:models}.
  Finally, we define $\int{\sstart}(\tau_1,\tau_2,\ldots,\tau_{n-1}) = 0^\omega$
  for all $\tau_1,\tau_2,\ldots,\tau_{n-1} \in \{0,1\}^\omega$,
  and $\int{A} = 1^a\, 0^\omega$.
  Then it is straightforward to verify that $\aalg$ is a model of the specification.
  
  For `$\Rightarrow$', let $\aalg = \pair{A}{\sint}$ be a model of the specification.
  Then we let the existential quantifiers
  $\exists \oracle_2,\exists \oracle_4,\ldots,\exists \oracle_n$ and $\exists x_2$ in~\eqref{eq:proof:full}
  behave according to the interpretations
  $\int{g_2},\int{g_4},\ldots,\int{g_n},\int{h_2}$, respectively
  (here the translation from sets $\oracle \subseteq \nat$ to streams $\num{\oracle}$ is as usual).
  Assume that there exists an assignment of the $\forall$ quantifiers $\forall \oracle_1,\forall \oracle_2,\ldots,\forall \oracle_{n-1}$
  and $\forall x_2$ for which the formula in~\eqref{eq:proof:full} is~not valid,
  that is, $M(\oracle_1,\ldots,\oracle_n,a,x_1,x_2)$ does not hold
  where the existential choices are governed by the model as described above.
  We translate this `counterexample' back to the model
  by considering $\int{\sstart}(\num{\oracle_1},\num{\oracle_3}\ldots,\num{\oracle_{n-1}})$.
  As in the proof of Theorem~\ref{thm:models}, it is then straightforward
  to show that $\int{\sstart}(\num{\oracle_1},\num{\oracle_3}\ldots,\num{\oracle_{n-1}}) \ne 0^\omega$.
  However, this contradicts the assumption of $\aalg$ being a model
  due to the equation $\sstart(\tau_1,\tau_3,\ldots,\tau_{n-1}) = \zeros$.
\end{proof}

The proof of Theorem~\ref{thm:models:full}
immediately yields the following:

\short{
  \begin{theorem}\label{thm:all:full}
    Each of the following problems (i), (ii), and (iii), 
    subsume the analytical hierarchy:
    \problem{Bitstream specification $\aes$, ground term $t \oftype \sbitstream$.}{%
      Does $t$ have: 
      (i)~a solution, (ii)~a unique solution, (iii)~at most one solution, 
      over all \,\emph{full}\, models of~$E$?}
  \end{theorem}
}
\extended{
  \begin{theorem}\label{thm:atleast:full}
    The following problem subsumes the analytical hierarchy:
    \problem{Bitstream specification $\aes$, term $s$.}{Does $s$ have a solution over all \emph{full}  models of $E$?}
  \end{theorem}
  \extended{%
  \begin{proof}
    Follows from the proof of Theorem~\ref{thm:models:full},
    as $\zeros$ has a solution over all models of $\aes$ 
    if and only if $\aes$ has a model.
  \end{proof}}
  
  \begin{theorem}\label{thm:unique:full}
    The following problem  subsumes the analytical hierarchy:
    \problem{Bitstream specification $\aes$, term $s$.}{Does $s$ have a unique solution over all \emph{full}  models of $E$?}
  \end{theorem}
  \extended{%
  \begin{proof}[Proof of Theorems~\ref{thm:atleast:full} and~\ref{thm:unique:full}]
    In the proof of Theorem~\ref{thm:models:full},
    $\zeros$ has a (unique) solution 
    if and only if $\aes$ has a model.
  \end{proof}
  }%
  
  For the proof of the following theorem, we slightly adapt the specification
  in the proof of Theorem~\ref{thm:models:full} such that it always has a solution,
  and has more than one solution if and only if the analytical formula in~\eqref{eq:proof:full} holds.
  
  \begin{theorem}\label{thm:atmost:full}
    The following problem subsumes the analytical hierarchy:
    \problem{Bitstream specification $\aes$, term $s$.}{Does $s$ have at most one solution over all \emph{full} models of $E$?}
  \end{theorem}
  \extended{%
  \begin{proof}
    We adapt the proof of Theorem~\ref{thm:models:full} by 
    exchanging the two equations $\sstart(\ldots) = \ldots$ by the following one:
    \begin{align*}
      &\ones = \lesseq{\sstart}{ \iszeros{ \srun(1,\: \zipn{n}{\tau_1,\msf{g_2}(\tau_1),\tau_3,\msf{g_4}(\tau_1,\tau_3),\\
        &\hspace{3cm}\ldots,\tau_{n-1},\msf{g_n}(\tau_1,\tau_3,\ldots,\tau_{n-1})},\: \zeros) }} 
    \end{align*}
    An interpretation $\int{\sstart} = 0^\omega$ always yields a solution.
    In addition, by Lemma~\ref{lem:lesseq} we have 
    $\int{\sstart} \ne 0^\omega$
    only if $\int{\iszeros{\ldots}} \ne 0^\omega$ for every assignment of $\tau_1,\tau_2,\ldots,\tau_{n-1}$.
    But then $\int{\iszeros{\ldots}} = 1^\omega$ by Lemma~\ref{lem:iszeros},
    and, thus, $\int{\srun(\ldots)} = 0^\omega$.
    As in the proof of Theorem~\ref{thm:models:full}, $\int{\srun(\ldots)} = 0^\omega$ for all
    $\tau_1,\tau_2,\ldots,\tau_{n-1}$ if and only if
    the formula in~\eqref{eq:proof:full} holds.
  \end{proof}
  }
}

\subsection{Equality of Solutions}

In this section, we study the complexity of deciding whether
terms have the same set of solutions over all (full) models.
It is easy to see that the hardness of these problems 
is at least that of deciding equality in all (full) models.
When considering all models, the problem turns out $\cpi{1}{2}$-complete,
and, thus, higher than the degree $\cpi{1}{1}$ of equality in all models.

\begin{remark}\label{rem:union}
  Let us briefly discuss the applicability of equality in~all (full) models
  for the comparison of terms $s$, $t$ that are specified in
  independent specifications $\aes_s$ and $\aes_t$.
  First, we rename the symbols of one of the specifications
  such that $\asig_s \cap \asig_t = \{{0,1,:}\}$.
  Thereafter, we consider the validity of $s = t$ in the union $\aes_s \cup \aes_t$. 

  We show on two examples that this approach does not always yield the intended results.
  Let $E_M$ consist of the single equation $M = 1:M$,
  and $E_N$ of 
  \begin{align*}
    N &= \inv{N} & \inv{0:\astr} &= 1:\inv{\astr} &
       \inv{1:\astr} &= 0:\inv{\astr}
  \end{align*}
  Then $M$ has the stream of ones as unique solution,
  but $N$ has no solution.
  Since $E_N$ does not have model, the
  union $E_M \cup E_N$ also does not admit one.
  Thus, $E_M \cup E_N \models M = N$ holds for trivial reasons.
  Nevertheless, we would not like to consider $M$ and $N$ as equivalent
  (at least if they are given by independent specifications).

  Even if the specifications have unique solutions,
  a similar~effect can occur. 
  Let $M = \zeros$ and $E_M$ consist of the equations
  \newcommand{\snxor}{\msf{nxor}}
  \newcommand{\nxor}{\funap{\msf{nxor}}}
  \begin{align*}
    \iszeros{\nxor{\astr}} &= \zeros\\
    \nxor{0:0:\astr} &= 1:\nxor{\astr} & \nxor{0:1:\astr} &= 0:\nxor{\astr}\\
    \nxor{1:0:\astr} &= 0:\nxor{\astr} & \nxor{1:1:\astr} &= 1:\nxor{\astr}
  \end{align*}
  together with the equations~\eqref{eq:filter}.
  Let $N = \alt$ and $E_N$ consist of the equation $\alt = 0:1:\alt$.
  Both specifications have models,
  and $\zeros$ and $\alt$ have unique solutions.
  For example, $E_M$ admits a model whose domain consists of
  all eventually constant streams.
  However, $E_M$ rules out models for which there exist elements
  $\astr \in A_{\sbitstream}$ with $\int{\snxor}(\astr) = 0^\omega$.
  In particular, the stream $0101\ldots$ is excluded from the domain $A_{\sbitstream}$.
  As a consequence, the union $E_M \cup E_N$ has no models,
  and $E_M \cup E_N \models \zeros = \alt$ holds.
\end{remark}

As a consequence of the proof of Theorem~\ref{thm:models:full}, we obtain:
\begin{theorem}\label{thm:solutions:full}
  The following problem  subsumes the analytical hierarchy:
  \problem{Bitstream specifications $\aes_s$, $\aes_t$, ground terms $s,t \,{\oftype}\, \sbitstream$.}{%
    Do $s$ and $t$ have equal solutions over all \emph{full} models,
    that is, $\solful{s}{\aes_s} = \solful{t}{\aes_t}$\,?}
\end{theorem}
\extended{
\begin{proof}
  Let $\aes_s$ be the specification in the proof of Theorem~\ref{thm:models:full}, and $s = \zeros$.
  Then $\solful{s}{\aes_s} = \{0^\omega\}$ if $\aes_s$ has a model, and $\setemp$ otherwise.
  Let $\aes_t = \{\zeros' = 0:\zeros'\}$ and $t' = \zeros'$,
  then we have $\solful{t}{\aes_t} = \{0^\omega\}$. 
  Thus, $\solful{s}{\aes_s} = \solful{t}{\aes_t}$
  is equivalent to $\aes \models \zeros = \ones$\, in the proof of Theorem~\ref{thm:models:full}.
\end{proof}
}

We conclude this section with an investigation of the complexity of deciding whether
two terms have the same set of solutions over all models.
The proof of Theorem~\ref{thm:models} yields only $\cpi{1}{1}$\nb-hardness.
In order to show $\cpi{1}{2}$-hardness, we employ a result of~\cite{omega:computations}
stating that it is a $\cpi{1}{2}$-complete problem to decide whether
the $\omega$-language of a non-deterministic Turing machine
contains all words $\{0,1\}^\omega$.

Therefore, we consider non-deterministic Turing machines with one\nb-sides tapes.
Without loss of generality, we may restrict the non-determinism
$\stmtrans \funin \autstates \times \tmsig \to \powerset{\tmstates \times \tmsig \times \{\tmL,\tmR\}}$
to binary choices in each step, that is, $\setsize{\tmtrans{\astate}{b}} \le 2$
for every $q \in \tmstates$ and $b \in \{0,1\}$.
(Broader choices then are simulated by sequences of binary choices.)
Moreover, for our purposes, it suffices to consider Turing machines that never halt.
For the $\omega$-language, halting always corresponds to rejecting a run,
and this rejection can be simulated by alternating moving forth and back eternally.

That is, a non-deterministic Turing machine 
$\tm = \quadruple{\tmstates}{\tmstart}{\stmtrans_0}{\stmtrans_1}$
has two transition functions
$\stmtrans_0,\stmtrans_1 \funin \autstates \times \tmsig \to \tmstates \times \tmsig \times \{\tmL,\tmR\}$
and we allow a non-deterministic choice between these~functions in each step.
Note that, for modeling non-determinism in an equational specifications,
we cannot take the union of the specifications
$\tmes{\triple{\tmstates}{\tmstart}{\stmtrans_0}}$
and
$\tmes{\triple{\tmstates}{\tmstart}{\stmtrans_1}}$,
since multiple equations having the same left-hand side
do not model choice, but additional restrictions on the models of the specification.
To this end, we introduce a third argument 
for the binary function symbols $q \in \tmstates$ in Definition~\ref{def:tmtrs}.
This argument then governs the non-deterministic choice.
In order to model one-sided tapes, 
we introduce a fourth argument that stores the position on the tape,
and is increased, when moving right, and decreased, when moving left.
That is, we adapt Definition~\ref{def:tmtrs} to:
  \begin{align*}
    \qfunap{\astate}{x}{b:y}{i:z}{p} &= \qfunap{\astate'}{b':x}{y}{z}{1:p} \\
    \qfunap{\astate}{a:x}{b:y}{i:z}{1:p} &= \qfunap{\astate'}{x}{a:b':y}{z}{p}
  \end{align*}
for $\stmtrans_i(\astate,b) = \triple{\astate'}{b'}{\tmR}$
and $\stmtrans_i(\astate,b) = \triple{\astate'}{b'}{\tmL}$, respectively.
We use $\tmesn{\tm}$ to denote this specification, and $\tmtrsn{\tm}$
for the corresponding term rewriting system.
In the initial configuration, the third argument should be an
underspecified stream, allowing for any non-deterministic choice.
We pass $\zeros$ as fourth argument, 
thereby ensuring that the head cannot move to negative tape indices.

A \emph{run} of $\tm$ on an $\omega$-word $w \in \{0,1\}^\omega$ 
is a $\tmtrsn{\tm}$ rewrite sequence starting from a term $\tmstart(\zeros,\underline{w},\underline{N},\zeros)$
where $N \in \{0,1\}^\omega$ determines the non-deterministic choices;
here $\underline{w}$ is the term $w(0) : w(1) : \ldots$
A run of $\tm$ is \emph{complete} if every tape position $p \ge 0$ is visited 
(that is, positions right of the starting position), and
it is \emph{oscillating} if some tape position is visited infinitely often.
A run is \emph{accepting} if it is complete and not oscillating, that it,
it visits every position $p \ge 0$ at least once, but only finitely often.

\begin{definition}
  The $\omega$-language $\lang{\tm}$
  is the set of all $\omega$-words $w \in \{0,1\}^\omega$ 
  such that $\tm$ has an accepting run $w$.
\end{definition}

We employ the following result, which follows from \cite{omega:computations}:
\begin{theorem}\label{thm:lang}
  The set $\{\tm \mid \lang{\tm} = \{0,1\}^\omega\}$ is $\cpi{1}{2}$-complete.
\end{theorem}

We are now ready for the proof of $\cpi{1}{2}$-completeness of
equality of the set of solutions over all models.
In the proof, we introduce a fifth argument for the
symbol $q \in \tmstates$ in $\tmesn{\tm}$ which 
enforces progress (productivity) and rules out exactly the oscillating runs.

\begin{theorem}\label{thm:solutions}
  The following problem is $\cpi{1}{2}$-complete:
  \problem{Bitstream specifications $\aes_s$, $\aes_t$, ground terms $s,t \,{\oftype}\, \sbitstream$.}{%
    Do $s$ and $t$ have equal solutions over all models equal,
    that is, $\sol{s}{\aes_s} = \sol{t}{\aes_t}$\,?}
\end{theorem}

\begin{proof}
  Let $\tm = \quadruple{\tmstates}{\tmstart}{\stmtrans_0}{\stmtrans_1}$ be a non-deterministic Turing~machine.
  We reduce the problem in Theorem~\ref{thm:lang} to 
  a decision problem for the equality of the set of solutions over all full models.
  We let $s = \msf{X}$ and define the specification $\aes_s$ to consist of:
  {\allowdisplaybreaks
  \begin{align}
    \tmstart(\zeros,\msf{X},\msf{N},\zeros,P) =\ &\zeros \label{eq:sol:start}\\
    \natstr{P} =\ &\ones \label{eq:sol:natstr} \\  
    \pfunap{\astate}{x}{b:y}{i:z}{p}{1:v} =\ &\pfunap{\astate'}{b':x}{y}{z}{1:p}{v} \label{eq:sol:tmr} \\
    &\text{for } \stmtrans_i(\astate,b) = \triple{\astate'}{b'}{\tmR}\hspace{.5cm} \notag\\
    \pfunap{\astate}{a:x}{b:y}{i:z}{1:p}{1:v} =\ &\pfunap{\astate'}{x}{a:b':y}{z}{p}{v} \label{eq:sol:tml} \\
    &\text{for } \stmtrans_i(\astate,b) = \triple{\astate'}{b'}{\tmL} \notag\\
    \pfunap{\astate}{x}{y}{z}{1:p}{0:v} =\  &0:\pfunap{\astate}{x}{y}{z}{p}{v} \label{eq:sol:prod}\\
    \pfunap{\astate}{x}{y}{z}{0:p}{0:v} =\  &\ones \label{eq:sol:fail1}\\
    \pfunap{\astate}{a:x}{b:y}{i:z}{0:p}{1:v} =\ &\ones \label{eq:sol:fail2}\\
    &\text{for } \stmtrans_i(\astate,b) = \triple{\astate'}{b'}{\tmL} \notag
  \end{align}}
  The equation~\eqref{eq:sol:start} starts $\tm$ on the stream $\msf{X}$
  with non-deterministic choices governed by $\msf{N}$ and
  $P$ for enforcing progress.
  The streams $\msf{X}$ and $\msf{N}$ are unspecified, thus arbitrary. 
  The equation~\eqref{eq:sol:natstr}
  ensures that $\int{P}$ contains infinitely many zeros.
  The equations~\eqref{eq:sol:tmr} and~\eqref{eq:sol:tml}
  model the computation of $\tm$ as discussed before,
  but now in each step removing the context $1:\hole$ from the fifth argument.
  If the fifth argument starts with a $0$, then \eqref{eq:sol:prod}
  decrements the position counter (the fourth argument).
  Recall, the position counter determines how many steps
  the Turing machine~$\tm$~is permitted to move left.
  Thus, always eventually decrementing the counter rules out the oscillating runs.
  The equations~\eqref{eq:sol:fail1} and~\eqref{eq:sol:fail2}
  rule out models where the head move left of the envisaged progress $\int{P}$.

  It is important to note that for any non-oscillating run $\sigma$,
  we can define a function $p : \nat \to \nat$ such that
  after $p(n)$ steps, $\tm$ visits only tape indices $\ge n$.
  Then an assignment $\int{P} = 1^{p(0)} \, 0\, 1^{p(1)} \,0\, 1^{p(2)} \,0 \ldots$
  in the model will permit this run to happen,
  that is, the head will never fall behind the envisaged progress
  and Equations~\eqref{eq:sol:fail1} and~\eqref{eq:sol:fail2} do not apply.

  As a consequence, we have $\sol{s}{\aes_s} = \{0,1\}^\omega$ if and only if
  for every $\int{\msf{X}} \in \{0,1\}^\omega$ there exists a non-oscillating run 
  (that is, an appropriate choice $\int{\msf{N}}$)
  of $\tm$ on $\int{\msf{X}}$.
  Now we define $t = \msf{Y}$ and $E_t = \{ \msf{Y} = \msf{Y} \}$
  for which obviously $\sol{t}{\aes_t} = \{0,1\}^\omega$.
  Therefore, $\sol{s}{\aes_s} = \sol{t}{\aes_t}$ if and only if $\lang{\tm} = \{0,1\}^\omega$.
  This concludes the proof of $\cpi{1}{2}$-hardness.
  
  For $\cpi{0}{2}$-membership, the problem can be characterized
  by the following analytical formula: 
  $\myall{\pair{\aalg_s}{\aalg_t}}{\myex{\pair{\aalg'_s}{\aalg'_t}}{
    (\aalg_s \models \aes_s \implies \aalg_t' \models \aes_t \wedge \int{s}^{\aalg_s} = \int{t}^{\aalg'_t}) \wedge
    (\aalg_t \models \aes_t \implies \aalg_s' \models \aes_s \wedge \int{t}^{\aalg_t} = \int{s}^{\aalg'_s})
  }}$.
  As in the proof of Theorem~\ref{thm:models}, here, it suffices to quantify over countable models.
\end{proof}

\section{Equality for Behavioral Specifications}\label{sec:rosu}
In this section we consider the notion of equality from~\cite{rosu:2006} which is based on hidden algebras~\cite{rosu:2000}.
We introduce the hidden models of bitstream specifications as employed in~\cite{rosu:2006},
where it has been shown that deciding the equality of (equationally defined) streams,
with respect to this semantics, is a $\cpi{0}{2}$\nb-complete problem.
We consider the following two extensions of this semantics:
\begin{enumerate}
  \item extending the semantics to streams over natural numbers, or
  \item requiring the behavioral equivalence $\equiv$ to be a congruence.
\end{enumerate}
We show that both extensions lift the complexity of deciding equality
to the level~$\cpi{1}{1}$ of the analytical hierarchy.
If the specifications are required to be productive 
(thus, separating the problem of productivity~\cite{endr:grab:hend:2009} from that of equality)
it can be shown that the complexity resides at $\cpi{0}{1}$~\cite{grab:endr:hend:klop:moss:2012}.
The results in~\cite{rosu:2006} (as well as the results we mention in the current paper)
are based on the comparison of non-productive specifications,
and the proofs inherently encode productivity problems.

Let us briefly explain why the $\cpi{1}{1}$-completeness for the equality of bitstreams in Theorem~\ref{thm:models}
does not directly carry over the setup of~\cite{rosu:2006}.
The problem is the definition of the function $\snatstr$ in~\eqref{eq:natstr}
containing the equation $\natstr{\ones} = \zeros$.
This equation does not work if we have confusion in the models 
and behavioral equivalence is not a congruence. In particular,
as discussed in Section~\ref{sec:related}, 
if $\ones' = 1:\ones'$, 
we cannot conclude that $\natstr{\ones'} = \zeros$.
As a consequence, with the behavioral specifications of \cite{rosu:2006}
it is not possible to enforce that a bitstream always eventually contains a zero.
However, if we consider behavioral specifications 
of streams of natural numbers,
then we no longer need $\snatstr$,
hence, reestablishing the $\cpi{1}{1}$-completeness result for the equality of streams of natural numbers
specified behaviorally.
There is a similar problem with the equation $\iszeros{\zeros} = \ones$, that, however, can be
overcome by discarding $\siszeros$ as in the proof of Theorem~\ref{thm:atleast}.

\subsection{Basic Setup}

In~\cite{rosu:2006}, every bitstream specification contains the equations
\begin{align*}
  \hd{x:\sigma} &= x &
  \tl{x:\sigma} &= \sigma
\end{align*}
where $\shd \oftype \sbitstream \to \sbit$ and $\stl \oftype \sbitstream \to \sbitstream$.

\begin{definition}
  A \emph{hidden $\asig$-algebra} $\aalg = \pair{A}{\sint}$ consists of
  \begin{enumerate}
    \item an $\sorts$-sorted domain $A$ where $A_{\sbit} = \{0,1\}$,
    \item for every $f \oftype s_1 \times \ldots \times s_n \to s \in \asig$ 
          an \emph{interpretation}
          $\int{f} : A_{s_1} \times \ldots A_{s_n} \to A_s$,
    \item $0,1 \in \asig$ with $\int{0} = 0$ and $\int{1} = 1$.
  \end{enumerate} 
\end{definition}
We stress that now $A_\sbitstream$ is an \emph{arbitrary} set.

\begin{definition}
  Let $\aalg = \pair{A}{\sint}$ be a hidden $\asig$-algebra.
  Then $\astr,\bstr \in A_{\sbitstream}$ are called \emph{behaviorally equivalent}, denoted by $\astr \equiv \bstr$,
  if they are indistinguishable with $\{\shd,\stl\}$-experiments, that is:
  \begin{align*}
    \astr \equiv \bstr \iff \myall{n \in \nat}{\int{\shd}(\int{\stl}^n(\astr) = \int{\shd}(\int{\stl}^n(\bstr)}
  \end{align*}
  On the domain $A_{\sbit}$, we let $\equiv$ be the identity relation. 
\end{definition}

Note that $\equiv$ is a not a congruence (only for $\int{\shd}$ and $\int{\stl}$).

\begin{definition}
  Let $\aes$ be a bitstream specification over $\asig$.
  A  hidden $\asig$-algebra \emph{behaviorally satisfies} $\aes$, denoted $\aalg \satisfies \aes$, if
  for every equation of $\aes$, the left- and right-hand sides are behaviorally equivalent:
  \short{%
    $\int{\ell}_{\alpha} \equiv \int{r}_{\alpha}$ 
    for every $\ell = r \in \aes$ and $\alpha : \avars \to A$.
  }%
  \extended{%
  \begin{align*}
    \int{\ell}_{\alpha} \equiv \int{r}_{\alpha} && \text{for every $\ell = r \in \aes$ and $\alpha : \avars \to A$}
  \end{align*}
  }%
  We say that an equation $\ell = r$ is \emph{behaviorally satisfied in all hidden models of $\aes$}, 
  denoted $\aes \satisfies \ell = r$
  if $\aalg \satisfies E$ implies $\aalg \satisfies \ell = r$
  for every hidden $\asig$-algebra $\aalg$.
\end{definition}

For a discussion of this semantics, we refer to Section~\ref{sec:related}.

\subsection{Behavioral Equivalence as Congruence}

We now adapt the basic setup by requiring $\equiv$ to be a \emph{congruence relation},
that is, $s \xbis t$ implies $f(\ldots,s,\ldots) \xbis f(\ldots,t,\ldots)$.
The resulting models are called \emph{behavioral} in~\cite{bido:henn:kurz:2003}.

\begin{definition}
  A hidden $\asig$-algebra is called \emph{behavioral} if $\equiv$
  is a congruence relation.
  For a bitstream specification $\aes$ over $\asig$, 
  we say that  \emph{$\ell = r$ is behaviorally satisfied in all behavioral models of $\aes$} if
  $\aalg \satisfies E \implies \aalg \satisfies \ell = r$ for every behavioral hidden $\asig$-algebra $\aalg$.
\end{definition}

\begin{theorem}\label{thm:behavioral:models}
  The following problem is $\cpi{1}{1}$-complete:
  \problem{Bitstream specification $\aes$, terms $s,t \oftype \sbitstream$.}{Is $s = t$  satisfied in all behavioral models of $E$?}
\end{theorem}

\begin{proof}
  We show: the equation $s = t$ is behaviorally satisfied in all behavioral models of $E$ 
  if and only if $s = t$ holds in all models of $E$;
  the latter property is $\cpi{1}{1}$-complete by Theorem~\ref{thm:models}.
  
  The direction `$\Leftarrow$' follows immediately,
  since every $\Sigma$-algebra is a behavioral hidden $\Sigma$-algebra.
  For `$\Rightarrow$', let $\aalg = \pair{A}{\sint}$ be a hidden $\Sigma$-algebra.
  Let $\aalg/_{\equiv} = \pair{A/_\equiv}{\sint/_\equiv}$ be the
  quotient~algebra. That is,
  $A/_\equiv$ are the congruence classes of $A$ with respect to $\equiv$. 
  For symbols $f \in \asig$ and $B_1,\ldots,B_{\ar{f}} \in A/_\equiv$,
  we define $\int{f}/_\equiv(B_1,\ldots,B_{\ar{f}}) = B$ if
  $\int{f}(b_1,\ldots,b_{\ar{f}}) = b$ for $b_1\in B_1,\ldots,b_{\ar{f}} \in B_{\ar{f}}$, 
  and $B$ is the congruence class of $b$ with respect to $\equiv$. 
  The quotient algebra $\aalg/_\equiv$ is a behavioral hidden $\asig$-algebra
  that, due to $\equiv$ being a congruence, behaviorally satisfies
  the same equations as $\aalg$.
  Let $\aalg'$ be the $\asig$\nb-algebra obtained from $\aalg/_\equiv$
  by renaming the domain elements into the streams they represent,
  that is, $a \in (\aalg/_\equiv)_{\sbitstream}$ becomes
  $\int{\shd}(a) : \int{\shd}(\int{\stl}(a)) : \ldots$.
  Then $\int{{:}}(x,\astr) = x:\astr$, since in $\aalg/_\equiv$
  every stream has a unique representative in the model.
  Hence, $\aalg'$ is a stream algebra.
  Moreover, for elements $a,b$ of the domain of $\aalg/_\equiv$,
  we have $a \equiv b$ iff $a = b$.
  Hence, $\aalg'$ is a model of an equation $s = t$ if and only if $s = t$ is behaviorally 
  satisfied in $\aalg$. 
\end{proof}

\subsection{Streams of Natural Numbers}

We briefly study hidden models with confusion, described in Section~\ref{sec:related}, for streams of natural numbers. 
A \emph{$\nat$-stream specification} is now defined like a bitstream specification,
except the sorts are $\sorts = \{\snat,\sstream\}$, and
the symbols are $0 \oftype \snat$, $s \oftype \snat \to \snat$ and `${:}$' of type $\snat \times \sstream \to \sstream$.
We adapt the definition of hidden $\asig$-algebras accordingly.
\begin{definition}
  A \emph{hidden $\asig$-algebra} $\aalg = \pair{A}{\sint}$ consists of
  \begin{enumerate}
    \item an $\sorts$-sorted domain $A$ and $A_{\snat} = \nat$,
    \item for every $f \oftype s_1 \times \ldots \times s_n \to s \in \asig$ 
          an \emph{interpretation}
          $\int{f} : A_{s_1} \times \ldots A_{s_n} \to A_s$,
    \item $0,s \in \asig$ with $\int{0} = 0$ and $\int{s}(x) = x+1$,
    \item for every $s \in A_{\sstream}$ there are $n \in \nat$  and $s' \in A_{\sstream}$ such that we have $s = \int{{:}}(b,s')$;
          see further Remark~\ref{rem:fix}.
  \end{enumerate} 
\end{definition}
The definitions of behavioral equivalence and satisfaction~are~the same as for bitstream specifications.
A slight modification of the proof of Theorem~\ref{thm:atleast} results in the following. 

\begin{theorem}\label{thm:confusion}
  The following problem  is $\cpi{1}{1}$-complete:
  \problem{$\nat$-stream specification $E$, terms $s,t  \oftype \sbitstream$.}{
    Does $\aes \satisfies s = t$ hold? 
    That is, is $s = t$ behaviorally satisfied in all hidden models of $\aes$?}%
\end{theorem}
\begin{proof}
  We reduce the well-foundedness problem for decidable binary relations to an equality problem.
  Let $M \subseteq \nat \times \nat$ be a decidable predicate,
  and $\tm = \triple{\tmstates}{\tmstart}{\stmtrans}$ the corresponding Turing machine.
  We define the following specification $\aes$:
  \begin{gather*}
    \begin{aligned}
    &\zeros = \srun(1,\msf{X}) &&&&\unary{0} = \zeros \\
    &\srun(0,\astr) = \ones    &&&&\unary{s(x)} = 1:\unary{x}\\
    \end{aligned}\\
    \begin{aligned}
    &\srun(1,\astr) = 0:\srun(\tminit{\zeros}{\unary{\hd{\astr}},\\
    &\hspace{3.5cm}\unary{\hd{\tl{\astr}}}},\tl{\astr})
    \end{aligned}
  \end{gather*}
  together with the equations from $\tmes{\tm}$ and~\eqref{eq:zap}.
  In contrast with the proof of Theorem~\ref{thm:atleast},
\pagebreak
  $\msf{X}$ is now a stream of natural numbers.
  Since $\msf{X}$ is unspecified, its interpretation in the model can be an arbitrary stream of natural numbers.
  As in the proofs of Theorems~\ref{thm:models} and~\ref{thm:atleast},
  we employ $\msf{X}$ to guess an infinite path through $M$.
  Instead of $\uhd{\cdot}$ and $\utl{\cdot}$ on bitstreams,
  we now take $\unary{\hd{\cdot}}$
  and $\tl{\cdot}$, respectively,
  where the function $\sunary$ converts natural numbers
  to unary representations in forms of streams.
  As in the proof of Theorem ~\ref{thm:atleast}, it follows that
  there exists a hidden $\asig$\nb-algebra $\aalg$ with $\aalg \satisfies \aes$
  if and only if $M$ is not well-founded.
  Thus, $\aes \satisfies \zeros = \ones$ if and only if $M$ is well-founded.
\end{proof}

\section{Equivalence of Lambda Terms}\label{sec:lambda}

\newcommand{\stdred}{\to_{\mit{std}}}
\newcommand{\cK}{\msf{K}}
\newcommand{\cI}{\msf{I}}
\renewcommand{\suc}{\msf{succ}}
\newcommand{\zer}{\msf{zer}}
\newcommand{\cbohm}{\funap{\msf{BT}}}
\newcommand{\clevi}{\funap{\msf{LT}}}
\newcommand{\sink}{\bot}
\newcommand{\hred}{\to_h}

In this section we investigate the complexity of deciding the equality
of $\lambda$-terms with respect to the observational equivalences $\seqn$, $\seqh$ and $\seqw$
as introduced in Section~\ref{sec:intro}.
Furthermore, we study the complexity of deciding whether
two $\lambda$-terms have the same \bohm{} trees or \levy{} trees.
The interested reader is referred to~\cite{bare:1984,deza:seve:vrie:2000}
for an introduction to \bohm{} trees, and to~\cite{deza:giov:2001} 
for a thorough study of the observational equivalences on $\lambda$-terms.

\begin{definition}\label{def:bohm} 
  Let $M$ be a $\lambda$-term.
  The \emph{\bohm{} tree $\cbohm{M}$ of $M$}
  is a potentially infinite term defined as follows.
  If $M$ has no hnf, then $\cbohm{M} = \sink$.
  Otherwise,
  there is a head reduction $M \hred^* \mylam{x_1}{\ldots\mylam{x_n}{y M_1 \ldots M_m}}$ to head normal form.
  Then we define 
  $\cbohm{M} = \mylam{x_1}{\ldots\mylam{x_n}{y \cbohm{M_1} \ldots \cbohm{M_m}}}$.
\end{definition}

\begin{definition}
  Let $M$ be a $\lambda$-term.
  The \emph{\levi{} tree $\clevi{M}$ of $M$}
  is a potentially infinite term defined as follows:
  \begin{align*}
    \clevi{M} & = \sink && \text{if $M$ has no whnf} \\
    \clevi{M} & = \mylam{x}{\clevi{N}} && \text{if $M \hred^* \mylam{x}{N}$} \\
    \clevi{M} & = x \clevi{M_1} \ldots \clevi{M_m} && \text{if $M \hred^* x M_1 \ldots M_m$}
  \end{align*}
%
\end{definition}

For the observational equivalences we obtain:
\begin{theorem}\label{thm:equiv}
  For each ${=_?} \in \{=_n,\, =_h,\, =_w\}$,
  the following problem is $\cpi{0}{2}$-complete:
  \problem{$\lambda$-terms $M$, $N$.}{Does $M =_? N$ hold?}
  \noindent
\end{theorem}

\begin{proof}
  First, we show $\cpi{0}{2}$-membership of the problem.
  We consider $=_n$ ($=_h$ and $=_w$ work analogously).
  A $\lambda$-term $Q$ has a normal form 
  if and only if $Q$ admits a standard reduction $\stdred^*$ to a normal form, see~\cite{bare:1984}.
For a $\lambda$-term $Q$, and $n\in\nat$, we write $Q \stdred^{\le n} \mit{nf}$
  to denote that $Q$ rewrites to a normal form within $\le n$ steps of standard reduction.
  Note that this is a decidable property.
  Then we claim:
  \begin{align}\label{frm:nf}
    M =_n N  & \iff  \\
     \myall{C}{ \myall{n} {\myex{m}{& \left(C[M]  \stdred^{\le n+m} \mit{nf} 
      \Leftrightarrow C[N] \stdred^{\le n+m} \mit{nf}\right)}}}    \nonumber   
  \end{align}
  For `$\Rightarrow$' in \eqref{frm:nf}, assume that $M =_n N$. Let $C$ be a context. 
  We distinguish the following cases:
  \begin{enumerate}
    \item \label{case:CM}
      Assume that $C[M]$ has a normal form. Then $C[N]$ has one,
      and $C[M] \stdred^k \mit{nf}$ and $C[N] \stdred^\ell \mit{nf}$ for some $k,\ell\in\nat$.
      Then in~\eqref{frm:nf} for any $n \in \nat$ we can choose $m = \max(k,\ell)$.
    \item The case that $C[N]$ has a normal form is symmetric to \ref{case:CM}.
    \item If neither $C[M]$ nor $C[N]$ have a normal form, then neither $C[M] \stdred^{\le n+m} \mit{nf}$
      nor $C[N] \stdred^{\le n+m} \mit{nf}$ for any $n,m\in\nat$.
  \end{enumerate}
  For `$\Leftarrow$' in \eqref{frm:nf}, assume $M \ne_n N$. 
  Then there is a context $C$ such that exactly one of the terms $C[M]$ and $C[N]$ has a normal form;
  without loss of generality, 
  assume $C[M] \stdred^{\le n} \mit{nf}$ for some $n \in\nat$.
  Hence, $C[M] \stdred^{\le n+m} \mit{nf}$ for every $m\in\nat$,
  but $C[N] \stdred^{\le n+m} \mit{nf}$ for no $m\in\nat$.
  Thus, the right-hand side of~\eqref{frm:nf} is not satisfied.
  
  From~\eqref{frm:nf} it follows that $=_n$ is in $\cpi{0}{2}$, 
  since the two quantifiers $\forall{C}$ and $\forall{n}$ can be merged into a single $\forall$-quantifier.
\pagebreak
  
  We now proceed with proving $\cpi{0}{2}$-hardness of the problem.
  Let $\msf{T}$ be a Turing machine, and let $T$ be a $\lambda$-term
  such that for all $n,m\in\nat$,
  $T\, \num{n}\, \num{m}$ rewrites to $\cK$ if $\msf{T}$ terminates on input $n$ within $m$ steps,
  and to $\cK\cI$, otherwise.
  Here, $\cK = \mylam{xy}{x}$ and $\cI = \mylam{x}{x}$
  are the usual combinators,
  and $\num{k} = \mylam{f}{\mylam{x}}{f^n x}$ is the Church numeral representing the natural number $k \in \nat$.
  The construction of such $T$ is standard, see \cite{bare:1984}.
  Now we define:
  \begin{align*}
    & \begin{aligned}
      M & = (\mylam{x}{\mylam{a}{a(xx)}})(\mylam{x}{\mylam{a}{a(xx)}})  &&&
       \\ 
    \end{aligned}\\    
    & \begin{aligned}
      N &= N'N'\zer &&&
      N' &= \mylam{xn}{T'n\,\zer (\mylam{a}{a(xx(\suc\,n))})} \\
      T' &= T''T''  &&&
      T'' &= \mylam{xnm}{Tnm\cI(xxn(\suc\, m))} \\ 
      \zer &= \mylam{fx}{x}   &&&
      \suc &= \mylam{zfx}{f (zfx)} \\ 
    \end{aligned}
  \end{align*}
  We show that $M =_? N$ if and only if $\msf{T}$ halts on all $n\in\nat$. 
  Note that 
  $T'\num{n}\,\num{m} \to^* \cI$ if $T\num{n}\,\num{m}\to^*\cK$, that is, 
  if $\msf{T}$ terminates on input $n$ in $m$ steps;
  otherwise $T'\num{n}\,\num{m} \to^* T'\num{n}\,(\num{m+1})$.
  Hence, we obtain
  \begin{align*}
   T'\num{n}\,\num{0} \to^* \cI 
    \iff & \text{$\msf{T}$ halts on input $n$} \\
   \iff & \text{$T'\num{n}\,\num{0}$ has a (weak) head normal form} 
  \end{align*}
  The \levy{} tree of $M$ is $\mylam{a}{a(\mylam{a}{a(\mylam{a}{a\ldots})})}$.
  If $\msf{T}$ halts on input $n$, we have 
  \[N'N'\num{n} \to^* T'\num{n}\,\zer(\mylam{a}{a(N'N'\num{n+1})}) \to^* \mylam{a}{a(N'N'\num{n+1})}\]
  Thus if $\msf{T}$ terminates on all $n\in\nat$,
  then the \levy{} trees of $M$ and $N$ are equal, and, hence, by \cite{deza:giov:2001}
  we have $M =_w N$, $M =_h N$ and $M =_n N$.
  Otherwise, let $n \in \nat$ be minimal such that $\msf{T}$ does not halt on $n$.
  Then by the above, we have:
  \[N \to^* \underbrace{\lambda a.\, a (\lambda a.\, a (\ldots \lambda a.\, a}_{\text{$n$-times}}(N'N'\num{n}) \ldots))\]
  Then $N\cI^n \to^* N'N'\num{n}$ has no (weak) head normal form, but $M\cI^n$ has.
  Thus we have $M \ne_n N$, $M \ne_h N$ and $M \ne_w N$.
  This proves $\cpi{0}{2}$-hardness.
\end{proof}

The proof immediately yields the following result:

\begin{theorem}\label{thm:lambda}
  The following problems are $\cpi{0}{2}$-complete:
  \problem{$\lambda$-terms $M$, $N$.}{
    \begin{enumerate}
      \item Do $s$ and $t$ have equal \bohm{} trees?
      \item Do $s$ and $t$ have equal \levi{} trees?
    \end{enumerate}
  }
\end{theorem}

\begin{proof}
  Follows immediately from the proof of Theorem~\ref{thm:equiv}
  since $M$ and $N$ are observationally equal if and only if they 
  have the same \levi{} tree, and for $M$ and $N$ the \levi{} trees
  coincide with their \bohm{} trees.
\end{proof}

We mention that for \ber{} trees, the proof of Theorem~\ref{thm:equiv}
implies $\cpi{0}{2}$-hardness. It is not difficult to see that
the problem of deciding the equality of \ber{} trees is in $\cpi{0}{3}$.
We leave the determination of the precise complexity to future work.

\section{Conclusions}
We have investigated different model-theoretic and rewriting based semantics of equality of
infinite objects, specified either by systems of equations or by $\lambda$-terms.
It turns out that the complexities for these notions vary
from the low levels of the arithmetical hierarchy $\cpi{0}{1}$ and $\cpi{0}{2}$,
up to $\cpi{1}{1}$ and $\cpi{1}{2}$ of the analytical hierarchy,
and some even subsume the entire arithmetical and analytical hierarchy.
In particular, the observational equivalences of $\lambda$\--terms,
that are of interest for functional programming,
are all $\cpi{0}{2}$-complete.

Apart from $\cpi{0}{1}$,
none of these classes are recursively enumerable or
co-recursively enumerable. Thus, there exists no complete proof systems for
proving or for disproving equality.
An exception is the equality of normal forms for productive specifications
for  which inequalities can be recursively enumerated~\cite{grab:endr:hend:klop:moss:2012}.
%

\bibliography{main}

\begin{thebibliography}{10}

\bibitem{bare:1984}
H.~P. Barendregt.
\newblock {\em {The Lambda Calculus, its Syntax and Semantics}}.
\newblock North-Holland, 1984.

\bibitem{bido:henn:kurz:2003}
M.~Bidoit, R.~Hennicker, and A.~Kurz.
\newblock {Observational Logic, Constructor-based Logic, and Their Duality}.
\newblock {\em Theor. Comput. Sci.}, 298:471--510, 2003.

\bibitem{buss:rosu:2000}
S.~R. Buss and G.~Rosu.
\newblock {Incompleteness of Behavioral Logics}.
\newblock {\em ENTCS}, 33:61--79, 2000.

\bibitem{omega:computations}
J.~Castro and F.~Cucker.
\newblock {Nondeterministic $\omega$-Computations and the Analytical
  Hierarchy}.
\newblock {\em Logik u. Grundlagen d. Math}, 35:333--342, 1989.

\bibitem{coqu:1993}
T.~Coquand.
\newblock {Infinite Objects in Type Theory}.
\newblock In {\em Postproc.\ Conf.\ on Types for Proofs and Programs
  (TYPES~1993)}, volume 806 of {\em LNCS}, pages 62--78. Springer, 1993.

\bibitem{deza:giov:2001}
M.~Dezani-Ciancaglini and E.~Giovannetti.
\newblock {From B\"ohm's Theorem to Observational Equivalences: an Informal
  Account}.
\newblock In {\em BOTH'01}, volume~50 of {\em ENTCS}, 2001.

\bibitem{deza:seve:vrie:2000}
M.~Dezani-Ciancaglini, P.~Severi, and F.-J. de~Vries.
\newblock B\"ohm's theorem for {B}erarducci trees.
\newblock In {\em {CATS 2000 Computing: the Australasian Theory Symposium}},
  volume~31 of {\em ENTCS}, 2000.

\bibitem{endr:geuv:simo:zant:2011}
J.~Endrullis, H.~Geuvers, J.~G. Simonsen, and H.~Zantema.
\newblock {Levels of Undecidability in Rewriting}.
\newblock {\em Information and Computation}, 209(2):227--245, 2011.

\bibitem{endr:grab:hend:2008}
J.~Endrullis, C.~Grabmayer, and D.~Hendriks.
\newblock {Data-Oblivious Stream Productivity}.
\newblock In {\em Proc.\ Conf.\ on Logic for Programming Artificial
  Intelligence and Reasoning (LPAR~2008)}, number 5330 in LNCS, pages 79--96.
  Springer, 2008.

\bibitem{endr:grab:hend:2009}
J.~Endrullis, C.~Grabmayer, and D.~Hendriks.
\newblock {Complexity of Fractran and Productivity}.
\newblock In {\em Proc.\ Conf.\ on Automated Deduction (CADE~22)}, volume 5663
  of {\em LNCS}, pages 371--387, 2009.

\bibitem{frie:wise:1976}
D.~P. Friedman and D.~S. Wise.
\newblock {CONS Should Not Evaluate its Arguments}.
\newblock In {\em ICALP}, pages 257--284, 1976.

\bibitem{geuv:1992}
H.~Geuvers.
\newblock {Inductive and Coinductive Types with Iteration and Recursion}.
\newblock In {\em Proc.\ Workshop on Types for Proofs and Programs
  (TYPES~1992)}, pages 193--217, 1992.

\bibitem{grab:endr:hend:klop:moss:2012}
C.~Grabmayer, J.~Endrullis, D.~Hendriks, J.~W. Klop, and L.~S. Moss.
\newblock {Automatic Sequences and Zip-Specifications}.
\newblock In {\em Proc.\ Symp.\ on Logic in Computer Science (LICS~2012)}. IEEE
  Computer Society, 2012.
\newblock To appear.

\bibitem{hend:morr::1976}
P.~Henderson and J.~H. Morris{,~}Jr.
\newblock {A Lazy Evaluator}.
\newblock In {\em Proc. ACM SIGACT-SIGPLAN Symp. on Principles on programming
  languages (POPL)}, pages 95--103. ACM, 1976.

\bibitem{malc:1997}
G.~Malcolm.
\newblock {Hidden Algebra and Systems of Abstract Machines}.
\newblock In {\em Proc.\ Symp.\ on New Models for Software Architecture
  (IMSA)}, 1997.

\bibitem{rosu:2000}
G.~Ro\c{s}u.
\newblock {\em {Hidden Logic}}.
\newblock PhD thesis, University of California, 2000.

\bibitem{rosu:2006}
G.~Ro\c{s}u.
\newblock {Equality of Streams is a {$\Pi^0_2$}-complete Problem}.
\newblock In {\em Proc. ACM SIGPLAN Conf. on Functional Programming (ICFP)},
  pages 184--191. ACM, 2006.

\bibitem{roge:1967}
H.~Rogers{,}~Jr.
\newblock {\em {Theory of Recursive Functions and Effective Computability}}.
\newblock McGraw-Hill, New York, 1967.

\bibitem{rutt:2003}
J.~J. M.~M. Rutten.
\newblock {Behavioural Differential Equations: a Coinductive Calculus of
  Streams, Automata, and Power Series}.
\newblock {\em Theor. Comput. Sci.}, 308(1-3):1--53, 2003.

\bibitem{rutt:2005b}
J.~J. M.~M. Rutten.
\newblock {A Tutorial on Coinductive Stream Calculus and Signal Flow Graphs}.
\newblock {\em Theor. Comput. Sci.}, 343:443--481, 2005.

\bibitem{sang:rutt:2012}
D.~Sangiorgi and J.~J. M.~M. Rutten.
\newblock {\em {Advanced Topics in Bisimulation and Coinduction}}.
\newblock Cambridge University Press, 2012.

\bibitem{shoe:1971}
J.~R. Shoenfield.
\newblock {\em {Degrees of Unsolvability}}.
\newblock North-Holland, 1971.

\bibitem{sijt:1989}
B.~A. Sijtsma.
\newblock {On the Productivity of Recursive List Definitions}.
\newblock {\em ACM Transactions on Programming Languages and Systems},
  11(4):633--649, 1989.

\bibitem{tere:2003}
Terese.
\newblock {\em {Term Rewriting Systems}}.
\newblock Cambridge University Press, 2003.

\bibitem{wali:meld:2001}
M.~Walicki and S.~Meldal.
\newblock Nondeterminism vs. underspecification.
\newblock In {\em Proc.\ of the World Multiconference on Systemics, Cybernetics
  and Informatics}, ISAS-SCI~2001, pages 551--555. IIIS, 2001.

\end{thebibliography}
\bibliographystyle{abbrv}

\end{document}